\documentclass[a4paper, amsfonts, amssymb, amsmath, reprint, showkeys, nofootinbib, twoside, aps, pra]{revtex4-2}
\usepackage[english]{babel}
\usepackage[utf8]{inputenc}
\usepackage[colorinlistoftodos, color=green!40, prependcaption]{todonotes}
\usepackage{amsthm}
\usepackage{mathtools}
\usepackage{physics}
\usepackage{xcolor}
\usepackage{graphicx}
\usepackage[left=23mm,right=13mm,top=35mm,columnsep=15pt]{geometry} 
\usepackage{adjustbox}
\usepackage{placeins}
\usepackage[T1]{fontenc}
\usepackage{lipsum}
\usepackage{csquotes}
\usepackage{caption}
\usepackage{amsthm}
\usepackage[font=small]{caption}

\usepackage{subcaption}
\allowdisplaybreaks
\theoremstyle{definition} 

\newtheorem{theorem}{Theorem}
\newtheorem{lemma}{Lemma}
\newtheorem{proposition}[theorem]{Proposition} 
\newtheorem{corollary}[theorem]{Corollary}
\newtheorem{conjecture}[theorem]{Conjecture}
\newtheorem{definition}{Definition}
\newtheorem{example}{Example}
\newtheorem{remark}{Remark}

\renewenvironment{proof}{\noindent{\bfseries Proof.}}{\hfill$\blacksquare$}

\usepackage[pdftex, pdftitle={Article}, pdfauthor={Author}]{hyperref}
\begin{document}
\title{Cumulant-based quantum relative Rényi functional}

\author{Atirat Meunson}
\thanks{Contributing author}
\email{atirat.meun@mail.kmutt.ac.th}
\affiliation{Quantum Computing and Information Research Centre (QX), 
Faculty of Science, King Mongkut’s University of Technology Thonburi, Thailand}

\author{Tanapat Deesuwan}
\thanks{Corresponding author}
\email{tanapat.dee@kmutt.ac.th}
\affiliation{Quantum Computing and Information Research Centre (QX), 
Faculty of Science, King Mongkut’s University of Technology Thonburi, Thailand}

\date{\today} 

\begin{abstract}
We introduce a new cumulant-based quantum relative Rényi functional as a candidate quantum Rényi divergence, derived from the cumulant-generating function (CGF) of the quantum relative surprisal operator and extending the classical connection between Rényi divergence and statistical cumulants to the quantum setting. Unlike the Petz and sandwiched quantum Rényi divergences, the proposed construction is motivated by statistical structure rather than operator-algebraic or operational principles. The functional naturally admits a path-integral-like representation through the Lie–Trotter product expansion, providing a trajectory-based interpretation of quantum divergence in Hilbert space. On its natural non-regularized domain for $\alpha>1$ under the support condition $\operatorname{supp}(\rho)\subseteq\operatorname{supp}(\sigma)$, we establish several fundamental properties, including positivity, reduction to the classical case, additivity, unitary invariance, continuity, and monotonicity with respect to the Rényi parameter $\alpha$. Whether the functional satisfies the quantum data-processing inequality (QDPI) under arbitrary CPTP maps remains open. To extend the analysis beyond the studied regime, we introduce a regularized version of the functional and study its behavior at $\alpha=0$. We show that the resulting relative quantumness quantity vanishes if and only if the underlying states commute, yielding a necessary and sufficient characterization of non-commutativity. For commutativity-preserving (CoP) channels, we further conjecture a QDPI-type monotonicity relation for this quantity. Extensive numerical simulations provide strong evidence in support of this conjecture, with no violations observed for the CoP channels considered in this work.

\end{abstract}

\maketitle

\section{Introduction} \label{intro}
The concept of information in physical and statistical systems is intimately tied to the notion of surprisal. In classical information
theory, it is conceptually fruitful to interpret the Shannon entropy not merely as a measure of uncertainty, but as the expectation value of surprisal or information content\cite{Shannon1948,ShannonWeaver1964,Wilde2013}. From this viewpoint, each possible outcome carries a degree of unexpectedness determined by its probability, and the Shannon entropy emerges as the first statistical moment of surprisal. This perspective naturally extends to relative information: the relative surprisal, expressed as the log-likelihood ratio, quantifies the additional
information required when data generated by one distribution are described as if they were drawn from another. Averaging this quantity yields the Shannon relative entropy, or Kullback-Leibler divergence\cite{Kullback1951,Cover2006}, a central measure of distinguishability in classical information theory.

While foundational, Shannon entropy and relative entropy capture only the mean (first moment) of surprisal, leaving its higher-order fluctuations unaddressed. Rényi entropy and Rényi divergence\cite{Renyi1961,vanErven2014} remedy this by introducing a one-parameter family of information measures indexed by $\alpha\in(0,+\infty)\setminus\{1\}$. These quantities probe higher-order statistical structure via generalized power means, allowing one to interpolate between sensitivity to typical and rare
events. Beyond their mathematical interest, Rényi-type quantities play a significant role across information theory and coding\cite{Bercher2009,AtarMerhav2015}, statistical mechanics and thermodynamics\cite{Fuentes2022,Baez2022}, and even black hole thermodynamics\cite{Czinner2016,Nakarachinda2025}.

In the quantum setting, probabilities are replaced by density operators, and the operator $-\ln\rho$ plays the role of quantum surprisal, encoding the intrinsic informational content of a quantum state. Correspondingly, the quantum relative surprisal between states $\rho$ and $\sigma$ is given by the operator difference $\ln\rho-\ln\sigma$\cite{Boes2022}. Its expectation value with respect to $\rho$ yields the Umegaki quantum relative entropy\cite{Umegaki1962}, 
which satisfies the quantum data processing inequality  (QDPI) under completely positive trace-preserving (CPTP) maps\cite{Lindblad1975,Hayashi2017} and plays a central role in quantum information theory.

In close analogy with the classical extension from Shannon to Rényi, quantum Rényi divergences further generalize the notion of quantum relative entropy so as to capture higher-order moments of the quantum relative surprisal operator. Two principal formulations have been developed in this context: the Petz quantum Rényi divergence\cite{Petz1986,Petz1986b} and the sandwiched quantum Rényi divergence\cite{muellerlennert2013,Wilde2014}, each motivated by distinct mathematical and operational considerations. These divergences have found broad applications, ranging from quantum thermodynamics and hypothesis testing to resource-theoretic formulations of information\cite{Wei2018,Guarnieri2019,Hiai2009, Li2024,Datta2009}.

Recently, Emori proposed a broad framework of quantum statistical functions, including quantum moment-generating,
characteristic, cumulant-generating, and second characteristic functions for quantum observables\cite{Emori2026}. This perspective is conceptually related to the cumulant-generating viewpoint adopted in the present work. However, the present work focuses specifically on the construction of a Rényi-type quantum divergence and its structural properties in the non-commutative setting.

At this point, we depart from the standard operator-algebraic constructions and adopt a statistical viewpoint centered on cumulants of the relative surprisal. Specifically, we introduce a cumulant-based quantum relative Rényi functional (hereafter referred to as the Cu-Q relative Rényi functional, pronounced “cue-cue”) $S^Q_\alpha(\rho\|\sigma)$. The construction is based on the cumulant-generating function (CGF) associated with the quantum relative surprisal operator. In this way, the classical connection between Rényi divergence and the CGF is extended to the non-commutative setting. In contrast to the Petz and sandwiched formulations, our functional admits a path-integral-like representation analogous to the Feynman formulation \cite{Feynman1965,Schulman1981}, enabling an interpretation of quantum divergence as a weighted sum over trajectories in Hilbert space.

The non-regularized Cu-Q relative Rényi functional is studied on the domain $\alpha>1,\quad \operatorname{supp}(\rho)\subseteq\operatorname{supp}(\sigma)$. The support condition guarantees that the logarithmic operators appearing in the definition are well-defined on the support of $\rho$. The restriction to $\alpha>1$ is motivated both by the parameter regime in which the fundamental quantum Rényi-type divergence properties of the proposed functional can be established and by the fact that, for non-full-rank states, the regime $0<\alpha<1$ may involve additional domain subtleties associated with the singular behavior of logarithmic operators on kernel subspaces. Moreover, the positivity proof developed in the present work relies essentially on the condition $\alpha-1>0$ and therefore does not extend directly to the interval $0<\alpha<1$. To avoid these technical complications, the present work focuses on the mathematically robust regime $\alpha>1$. On this domain, the proposed functional satisfies several essential properties, including positivity, additivity, unitary invariance, continuity, and monotonicity with respect to the Rényi parameter, indicating behavior consistent with established quantum divergence measures.

For analytical purposes, we additionally introduce a regularized version of Cu-Q relative Rényi functional
$S_{\alpha}^{Q}(\rho_\varepsilon\|\sigma_\varepsilon)$,
constructed from full-rank approximations of the density operators. This regularized formulation allows the functional to be considered for all $\alpha\in\mathbb{R}\setminus\{1\}$. This regularized framework further enables the construction of a relative quantumness functional $Q(\rho_\varepsilon\|\sigma_\varepsilon):=- S_0^Q(\rho_\varepsilon\|\sigma_\varepsilon)$, which vanishes if and only if the states commute, thereby providing a necessary and sufficient characterization of non-commutativity. Building on the notion of commutativity-preserving (CoP) channels introduced in\cite{Yu2011}, we further formulate a conjecture that this relative quantumness functional satisfies a data-processing inequality under such channels for non-commuting input states. The conjecture is supported by numerical simulations on randomly generated qubit and qutrit states across several representative CoP channels, with no violations observed. More generally, whether the Cu-Q relative Rényi functional satisfies the quantum data-processing inequality under arbitrary CPTP maps remains an open and nontrivial problem. While an analytic proof is currently unavailable, these results provide new structural insight into the data-processing behavior of the proposed functional and identify a natural class of channels for which progress appears possible.

A preliminary version of this work appeared in
Ref.\cite{MeunsonDeesuwan2025CQRRF}. The present article extends that study through a more complete theoretical formulation, deeper structural analysis, and new analytical and numerical results concerning the regularized boundary quantity. Rather than surveying the entire landscape of quantum Rényi divergences, we focus on aspects directly relevant to the cumulant-based formulation developed here. In particular, applications to quantum resource theories and quantum thermodynamics are beyond the scope of the present work.
 
\section{Mathematical framework} \label{Math}
Throughout this paper, we consider a finite-dimensional Hilbert space $\mathcal{H}$ with $\dim \mathcal{H} = d$ over $\mathbb{C}$. We denote by $\mathcal{B}(\mathcal{H})$ the algebra of all linear bounded operators on $\mathcal{H}$, and by
\begin{equation}
    \mathcal{B}_h(\mathcal{H}) := \left\{ A \in \mathcal{B}(\mathcal{H}) \mid A = A^{\dagger} \right\}
\end{equation}
the real vector space of Hermitian operators. 
The set of positive semidefinite (PSD) operators is defined as
\begin{equation}
    \mathcal{P}(\mathcal{H}) := \left\{ A \in \mathcal{B}_h(\mathcal{H}) \mid A \ge 0 \right\},
\end{equation}
while the set of positive definite (PD) operators is
\begin{equation}
    \mathcal{P}_+(\mathcal{H}) := \left\{ A \in \mathcal{B}_h(\mathcal{H}) \mid A > 0 \right\}.
\end{equation}
We consider pairs of quantum states $(\rho,\sigma)$, where each state is represented by a density operator on $\mathcal H$. 
We denote the set of such pairs by
\begin{multline}
    \mathcal{D}(\mathcal{H}) \times \mathcal{D}(\mathcal{H}) :=\\
    \left\{ (\rho,\sigma) \mid \rho,\sigma \in \mathcal{P}(\mathcal{H}),\ 
    \operatorname{Tr}(\rho)=\operatorname{Tr}(\sigma)=1 \right\}.
\end{multline}
Similarly, we define
\begin{multline}
    \mathcal{D}_+(\mathcal{H}) \times \mathcal{D}_+(\mathcal{H}) \\
    := \left\{ (\rho,\sigma) \mid \rho,\sigma \in \mathcal{P}_+(\mathcal{H}),\ 
    \operatorname{Tr}(\rho)=\operatorname{Tr}(\sigma)=1 \right\}.
\end{multline}
Clearly, the set of full-rank density operators $\mathcal D_+(\mathcal H)$ 
forms a proper subset of the set of all density operators $\mathcal D(\mathcal H)$. 
Indeed, $\mathcal D_+(\mathcal H)\subseteq \mathcal D(\mathcal H)\subseteq \mathcal B(\mathcal H)$. For any pair $(\rho,\sigma)\in \mathcal{D}(\mathcal{H}) \times \mathcal{D}(\mathcal{H})$,
the spectral theorem ensures that $\rho$ and $\sigma$ admit spectral decompositions.
In particular,
\begin{equation}\label{rho}
    \rho = \sum_{i=1}^d p(\varphi_i)\, |\varphi_i \rangle \langle \varphi_i |, \quad
    \sigma = \sum_{j=1}^d q(\psi_j)\, |\psi_j \rangle \langle \psi_j |,
\end{equation}
where $p(\varphi_i), q(\psi_j) \ge 0$ and $\sum_{i=1}^d p(\varphi_i) = \sum_{j=1}^d q(\psi_j) = 1$.

If $(\rho,\sigma)\in \mathcal{D}_+(\mathcal{H})\times\mathcal{D}_+(\mathcal{H})$,
then all eigenvalues are strictly positive, i.e.,
$p(\varphi_i)>0$ and $q(\psi_j)>0$ for all $i,j$, and both $\rho$ and $\sigma$ are invertible. In general, the operators $\rho$ and $\sigma$ do not commute and therefore cannot be simultaneously diagonalized. 

For a positive semidefinite operator $A \in \mathcal P(\mathcal H)$, its support is defined as the orthogonal complement of its kernel, i.e.,
The support of $A$ is defined as $\operatorname{supp}(A) := \ker(A)^{\perp}$. Equivalently, if $A$ admits a spectral decomposition
\begin{align}
    A &= \sum_i a(k_i) \, |k_i\rangle\langle k_i|, \\
    \text{then } \operatorname{supp}(A) &= \operatorname{span}\left\{ |k_i\rangle \in \mathcal{H} : a(k_i) > 0 \right\}.
\end{align}

In this work, we impose the support condition $\operatorname{supp}(\rho) \subseteq \operatorname{supp}(\sigma)$, which is standard in the study of relative entropic quantities. This condition ensures that operator functions such as $\ln\rho$ and $\ln\sigma$
are well defined on the relevant subspace and prevent divergences arising from zero eigenvalues of $\sigma$.

Since the Hilbert space $\mathcal{H}$ considered in this work is finite-dimensional, every operator in $\mathcal B(\mathcal H)$ is compact.
For a compact operator $A \in \mathcal B(\mathcal H)$, let $\{ s_n(A) \}_{n \ge 1}$ denote the
singular values of $A$, i.e. the eigenvalues of $|A| := (A^{\dagger}A)^{1/2}$,
arranged in non-increasing order and counted with multiplicity.
For $1 \le p < \infty$, the Schatten $p$-class $\mathcal{L}_p(\mathcal H)$ is defined as
\begin{equation}
\begin{split}
    \mathcal{L}_p(\mathcal H) := \Big\{ A \in \mathcal B(\mathcal H) \ \Big|\ & A \text{ is compact and} \\
    &\sum_{n=1}^{d} s_n(A)^p < \infty \Big\},
\end{split}
\end{equation}
equipped with the norm
\begin{equation}
    \|A\|_p := \left( \sum_{n=1}^{d} s_n(A)^p \right)^{1/p}.
\end{equation}

For $p = \infty$, we set $\mathcal{L}_\infty(\mathcal{H}) = \mathcal{B}(\mathcal{H})$ endowed with the operator norm. This is consistent with the limit
\begin{equation}
    \|A\|_\infty = \lim_{p\to \infty} \|A\|_p,
\end{equation}
where the operator norm $\|A\|_\infty$ is defined as
\begin{equation}
    \|A\|_\infty := \sup_{\substack{x \in \mathcal{H} \\ \|x\|=1}} \|Ax\|.
\end{equation}
where the supremum is taken over all unit vectors in the Hilbert space.
The class $\mathcal{L}_1(\mathcal{H})$ coincides with the space of trace-class operators, while
$\mathcal{L}_2(\mathcal{H})$ coincides with the Hilbert-Schmidt class.
It is a standard result that $\mathcal{L}_1(\mathcal{H})$ is a two-sided ideal in $\mathcal{B}(\mathcal{H})$. 
Specifically, if $A \in \mathcal{B}(\mathcal{H})$ and $T \in \mathcal{L}_1(\mathcal{H})$, then both products 
$AT$ and $TA$ belong to $\mathcal{L}_1(\mathcal{H})$ and satisfy the norm bounds
\begin{equation}
    \|AT\|_1 \le \|A\|_\infty \|T\|_1 \qquad \text{and} \qquad \|TA\|_1 \le \|A\|_\infty \|T\|_1.
\end{equation}

Moreover, the trace functional extends to a continuous bilinear pairing
between $\mathcal{L}_1(\mathcal{H})$ and $\mathcal{B}(\mathcal{H})$, given by $(T,A) \longmapsto \operatorname{Tr}(TA)$. 
For $1 \le p \le \infty$ with conjugate exponents $p$ and $q$ satisfying $1/p + 1/q = 1$, 
the Schatten--Hölder inequality holds:
\begin{equation}
    |\operatorname{Tr}(AB)| \le \|A\|_p \|B\|_q, \quad A \in \mathcal{L}_p(\mathcal{H}), \; B \in \mathcal{L}_q(\mathcal{H}).
\end{equation}
In particular, for the boundary case $p=1$ and $q=\infty$, one obtains
\begin{equation}
    |\operatorname{Tr}(AB)| \le \|A\|_1 \|B\|_\infty,
\end{equation}
which reflects the canonical Banach space duality $\mathcal{L}_1(\mathcal{H})^* = \mathcal{B}(\mathcal{H})$ via 
$\langle A,B\rangle = \operatorname{Tr}(AB)$. This duality is a standard structural theorem in operator theory 
(see, e.g., \cite[Theorem VI.26]{ReedSimon1980}, \cite[Theorem 1.4]{RecentSchwarzTraceIneq}, \cite{Kittaneh1985}, and \cite{Watrous2018}).

\section{Preliminaries} \label{pre}
\subsection{Classical information entropy}

In classical information theory, uncertainty and distinguishability between probability distributions are quantified by entropy and divergence measures. 
Let $X : \Omega \to \mathcal{X}$ be a discrete random variable defined on a probability space $(\Omega,\mathcal{F},\mathbb{P})$, with probability mass function $p(x)=\mathbb{P}(X=x)$ on a finite or countable alphabet $\mathcal{X}$. 
The \emph{surprisal} associated with an outcome $x \in \mathcal{X}$ is defined as
\begin{equation}\label{surprisal}
\Lambda(x) := -\ln p(x),
\end{equation}
which quantifies the information gained upon observing $x$. 
Throughout this work, we adopt the natural logarithm, following conventions in quantum information theory.

The expectation value of the surprisal with respect to $p(x)$ yields the Shannon entropy~\cite{Shannon1948,ShannonWeaver1964},
\begin{equation}
S(X) := \mathbb{E}[\Lambda(X)]
     = -\sum_{x\in\mathcal{X}} p(x)\ln p(x),
\end{equation}
which represents the average information content, or uncertainty, of the random variable $X$.

To quantify the distinguishability between two probability distributions, let $p(x)$ denote the true distribution and $q(x)$ a reference distribution defined on the same alphabet, with $\mathrm{supp}(p)\subseteq\mathrm{supp}(q)$. 
The corresponding relative surprisal is given by
\begin{equation}\label{relative surprisal}
\Delta\Lambda(x) := \ln p(x)-\ln q(x),
\end{equation}
and its expectation with respect to $p(x)$ defines the Shannon relative entropy, or Kullback--Leibler divergence~\cite{Kullback1951},
\begin{equation}
S(p\| q)
= \sum_{x\in\mathcal{X}} p(x)\ln\frac{p(x)}{q(x)}.
\end{equation}
This quantity measures the average excess surprisal incurred when events generated according to $p$ are described using $q$.

While Shannon entropy and relative entropy capture only the first moment of the surprisal, Rényi introduced a one-parameter family of generalized entropies that probe higher-order statistical features of the information content~\cite{Renyi1961}. 
The Rényi entropy of order $\alpha \in (0,+\infty)\setminus\{1\}$ is defined as
\begin{equation}\label{Renyi}
S_\alpha(X)
= \frac{1}{1-\alpha}
  \ln\!\left(
  \sum_{x\in\mathcal{X}} p(x)^{\alpha}
  \right),
\end{equation}
which smoothly recovers the Shannon entropy in the limit $\alpha\to1$.
By tuning the parameter $\alpha$, the entropy becomes more sensitive to either rare or typical events, thereby extending the Shannon framework.

An analogous generalization applies to relative entropy. 
For two distributions $p(x)$ and $q(x)$ defined on the same alphabet, the Rényi divergence of order $\alpha$ is given by
\begin{equation}
S_\alpha(p\| q)
= \frac{1}{\alpha-1}
  \ln\!\left(
  \sum_{x\in\mathcal{X}} p(x)^{\alpha} q(x)^{1-\alpha}
  \right),
\end{equation}
which reduces to the Kullback--Leibler divergence in the limit $\alpha\to1$. 
This quantity interpolates between different notions of distinguishability and provides a natural extension of relative entropy beyond the level of average surprisal.

\subsection{Quantum information entropy}

The transition from classical to quantum information theory arises when probability
distributions are replaced by density operators acting on a Hilbert space.
In this setting, uncertainty and distinguishability are described by operator-valued
generalizations of classical entropy and divergence measures.

Let $\rho = \sum_i p(\varphi_i) |\varphi_i \rangle \langle \varphi_i| \in \mathcal{D}(\mathcal{H})$
be a quantum state on a finite-dimensional Hilbert space $\mathcal{H}$, where
$\{|\varphi_i\rangle\}$ forms an orthonormal eigenbasis of $\rho$.
The quantum analogue of the Shannon entropy is the von Neumann entropy,
defined as
\begin{equation}\label{QE}
S(\rho) = -\mathrm{Tr}\,\rho \ln \rho.
\end{equation}
The operator $\Xi := -\ln \rho$ is referred to as the quantum surprisal, encoding the self-information content of the state.
When $\rho$ is diagonal in an orthonormal basis representing classical outcomes, $S(\rho)$ reduces to the Shannon entropy of the associated probability distribution.

While the von Neumann entropy quantifies the uncertainty of a single quantum state,
it does not capture the distinguishability between two states.
This motivates the introduction of an operator-valued analogue of the classical
relative surprisal.

For two density operators $\rho, \sigma \in \mathcal{D}(\mathcal{H})$ acting on the same Hilbert space, the operator-valued analogue of the classical relative surprisal is given by the logarithmic difference, which we refer to as the quantum relative surprisal~\cite{Boes2022}
\begin{equation}\label{QRsurprisal}
\Delta \Xi := \ln \rho - \ln \sigma.
\end{equation}
This quantity captures the relative self-information of $\rho$ with respect to $\sigma$
at the operator level.
When $\mathrm{supp}(\rho) \subseteq \mathrm{supp}(\sigma)$, taking the expectation value
of $\Delta \Xi$ with respect to $\rho$ yields the Umegaki quantum relative entropy,
\begin{equation}\label{QRE}
S(\rho \| \sigma)
=\mathrm{Tr}\,\rho(\ln \rho - \ln \sigma),
\end{equation}
which quantifies the statistical distinguishability between quantum states.
Unlike its classical counterpart, the relative surprisal operator $\Delta \Xi$ is
generally non-commutative, reflecting an intrinsically quantum feature that has no
direct classical analogue and plays a central role in the structure of quantum
divergence measures.

In direct analogy with the classical theory, the Rényi entropy provides a one-parameter
generalization of the von Neumann entropy, parametrized by
$\alpha \in (0,+\infty)\setminus\{1\}$.
For a quantum state $\rho \in \mathcal{D}(\mathcal{H})$, it is defined as
\begin{equation}
S_{\alpha}(\rho)
=
\frac{1}{1-\alpha}
\ln \mathrm{Tr}\!\left(\rho^{\alpha}\right),
\end{equation}
and satisfies the Rényi entropy axioms for a single quantum state~\cite{muellerlennert2013}.

Extending this framework from single states to pairs of states leads naturally to a
family of quantum Rényi divergences, which generalize Umegaki quantum relative entropy.
In the commuting case, where $[\rho,\sigma]=0$, these divergences reduce to their
classical counterparts. In the generic non-commuting case, however, the absence of simultaneous diagonalization
introduces genuinely quantum effects and additional mathematical subtleties.

Among the most prominent constructions in this setting are the Petz quantum Rényi
divergence and the sandwiched quantum Rényi divergence, which emphasize different
operator-theoretic structures and operational behaviors.

\begin{definition}[Petz quantum Rényi divergence]
    Let $\rho$ and $\sigma$ be density operators in $\mathcal{D}(\mathcal{H})$. For $\alpha \in (0,+\infty)\setminus\{1\}$, the Petz quantum Rényi divergence is defined as
    \begin{equation}
        S_{\alpha}^{\mathrm{Petz}}(\rho \parallel \sigma)
        :=
        \begin{cases}
            \dfrac{1}{\alpha - 1} \ln \operatorname{Tr}(\rho^{\alpha}\sigma^{1-\alpha}),
            & \text{if } (\rho, \sigma) \in \mathcal{R}, \\
            +\infty, & \text{otherwise,}
        \end{cases}
        \label{eq:petz}
    \end{equation}
    where the admissible set $\mathcal{R}$ is defined explicitly by
    \[
        \mathcal{R} := \left\{ (\rho,\sigma) \in \mathcal{D}(\mathcal{H}) \times \mathcal{D}(\mathcal{H})
        \ \middle|\ \begin{aligned} &\rho \not\perp \sigma \text{ and} \\ &\operatorname{supp}(\rho) \subseteq \operatorname{supp}(\sigma) \end{aligned} \right\}.
    \]
\end{definition}
This construction, rooted in operator convexity and modular theory, preserves
positivity and admits a well-defined classical limit~\cite{Petz1986,Petz1986b}.

To strengthen monotonicity properties under quantum channels, a symmetrized construction known as the sandwiched quantum Rényi divergence was proposed.

\begin{definition}[Sandwiched quantum Rényi divergence]
   Let $\rho$ and $\sigma$ be density operators in $\mathcal{D}(\mathcal{H})$. For $\alpha \in (0,+\infty)\setminus\{1\}$, the sandwiched quantum Rényi divergence is defined under the same domain restrictions as
\begin{equation}
\begin{split}
    &S_{\alpha}^{\mathrm{s}}(\rho \parallel \sigma) \\
    &:= \begin{cases} 
        \dfrac{1}{\alpha - 1} \ln \operatorname{Tr}\!\left[ \left( \sigma^{\frac{1-\alpha}{2\alpha}} \rho \sigma^{\frac{1-\alpha}{2\alpha}} \right)^{\!\alpha} \right], & \text{if } (\rho, \sigma) \in \mathcal{R}, \\
        +\infty, & \text{otherwise.}
    \end{cases}
\end{split}
\label{eq:sandwiched}
\end{equation}

\end{definition}
A unifying two-parameter extension is given by the $(\alpha,z)$-quantum Rényi divergence,
which interpolates between the Petz and sandwiched forms \cite{audenaert2015alphaZ}.

The quantum Rényi divergences are of particular interest due to their monotonicity under CPTP maps, a property known as the QDPI. Specifically, for $(\rho,\sigma )\in \mathcal{D}(\mathcal{H}) \times \mathcal{D}(\mathcal{H})$ with $\mathrm{supp}(\rho)\subseteq\mathrm{supp}(\sigma)$, any CPTP map $\mathcal{N}_{CPTP}$ satisfies:
\begin{equation}
    S_{\alpha}(\rho \| \sigma) \ge S_{\alpha}\bigl(\mathcal{N}_{CPTP}(\rho)\| \mathcal{N}_{CPTP}(\sigma)\bigr),
\end{equation}
under the following parameter regimes:
\begin{itemize}
    \item  {Petz quantum Rényi divergence:} $S_{\alpha}^{\mathrm{Petz}}(\rho \| \sigma)$ satisfies QDPI for $\alpha\in(0,2]$ \cite{Tomamichel2016,Li2024}.
    \item {Sandwiched quantum Rényi divergence:} $S_{\alpha}^{\mathrm{s}}(\rho \|\sigma)$ satisfies QDPI for $\alpha\ge \tfrac{1}{2}$ \cite{frank2013,Beigi2013,Wilde2014}.
    \item {$(\alpha,z)$-quantum Rényi divergence:} $S_{\alpha,z}(\rho \| \sigma)$ satisfies QDPI in broad regions including the Petz ($z=1$) and sandwiched ($z=\alpha$) cases \cite{audenaert2015alphaZ}.
\end{itemize}

These results indicate that, despite their distinct algebraic constructions, existing quantum Rényi divergences share monotonicity under CPTP maps within specific parameter regimes. In the subsequent sections, we investigate whether a cumulant-based formulation of quantum Rényi divergence admits a comparable data-processing behavior, motivated by the statistical structure of the relative surprisal operator.
\section{The Cumulant-Generating Quantum Rényi Functional}

In the classical setting, Rényi divergences can be expressed as logarithmic moment-generating functions of the relative surprisal. Motivated by this representation, we extend the construction to the quantum setting by replacing the classical relative surprisal with the quantum relative surprisal operator $\Delta \Xi$, thereby obtaining a non-commutative analogue of the cumulant-generating functional.

\begin{definition}[Cumulant-based quantum relative Rényi functional] 
    Let $\rho$ and $\sigma$ be density operators in $\mathcal{D}(\mathcal{H})$. For $\alpha >1$, the cumulant-based quantum relative Rényi functional (Cu-Q relative Rényi functional) is defined as
    \begin{equation}
    \label{QCRRF}
    \begin{split}
        &S_\alpha^{\text{Q}}(\rho \| \sigma) \\
        &:= \begin{cases} 
            \dfrac{1}{\alpha - 1} \ln \operatorname{Tr} \left( \rho e^{(\alpha - 1) (\ln \rho - \ln \sigma)} \right), & \text{if } (\rho, \sigma) \in \mathcal{R}, \\ 
            +\infty, & \text{otherwise.} 
        \end{cases} 
    \end{split}
    \end{equation}
\end{definition}

The non-regularized Cu-Q relative Rényi functional constitutes the primary object of study in this work. Throughout the paper, all structural properties of the Cu-Q relative Rényi functional are established for its studied regime, namely $\alpha>1,\operatorname{supp}(\rho)\subseteq
\operatorname{supp}(\sigma)$.

To investigate the behavior of the Cu-Q relative Rényi functional outside its studied regime, we introduce a regularization procedure based on strictly
positive approximations of the density operators. This construction
extends the Cu-Q relative Rényi framework beyond its studied regime and
ensures that the logarithmic operators appearing in the definition are
well defined on the entire Hilbert space.

Given $\rho,\sigma\in\mathcal D(\mathcal H)$ and a parameter
$\varepsilon\in(0,1)$, define
\begin{equation}\label{state_regularized}
\rho_\varepsilon
:=
(1-\varepsilon)\rho
+
\varepsilon\frac{\mathbb{I}}{d},
\quad
\sigma_\varepsilon
:=
(1-\varepsilon)\sigma
+
\varepsilon\frac{\mathbb{I}}{d},
\end{equation}
where $d=\dim(\mathcal H)$. Since
\begin{equation}
\rho_\varepsilon
\ge
\frac{\varepsilon}{d}{\mathbb{I}},
\quad
\sigma_\varepsilon
\ge
\frac{\varepsilon}{d}{\mathbb{I}},
\end{equation}
both regularized states are strictly positive and therefore belong to
$\mathcal D_+(\mathcal H)$. Consequently,
$\ln\rho_\varepsilon$ and $\ln\sigma_\varepsilon$ are well defined on the entire Hilbert space. The regularized Cu-Q relative Rényi functional is defined,
for $\alpha\in\mathbb R\setminus\{1\}$, by
\begin{equation}\label{regularized_Cu-Q}
S^{Q}_{\alpha}(\rho_\varepsilon\|\sigma_\varepsilon)
:=
\frac{1}{\alpha-1}
\ln
\operatorname{Tr}
\!\left[
\rho_\varepsilon
e^{(\alpha-1)
(\ln\rho_\varepsilon-\ln\sigma_\varepsilon)}
\right].
\end{equation}

For $\alpha>1$ and
$\operatorname{supp}(\rho)\subseteq\operatorname{supp}(\sigma)$, the regularized construction recovers the original Cu-Q relative Rényi functional in the limit $\varepsilon\to0^+$. The regularized framework does not alter the underlying structure of the proposed functional; rather, it provides a mathematically convenient extension to strictly positive density operators, allowing the same class of
properties to be investigated on the larger parameter domain $\alpha\in\mathbb{R}\setminus\{1\}$. 

To make explicit the cumulant-generating-function structure underlying the proposed construction, we recall the quantum relative surprisal
operator introduced in Eq.\eqref{QRsurprisal}. We then define 
\begin{equation}
\label{eq:CGF}
    \mathcal{F}_{\rho,\sigma}(\tau)
    :=
    \ln\operatorname{Tr}\!\left[
        \rho\, e^{\tau\Delta\Xi}
    \right],
    \quad \tau\in\mathbb{R},
\end{equation}
where $e^{\tau\Delta\Xi}$ is defined via the functional calculus of $\Delta\Xi$ on $\operatorname{supp}(\rho)$.

Since $\operatorname{supp}(\rho)$ is finite-dimensional, the spectral theorem yields
\begin{equation}
    \Delta\Xi = \sum_j \lambda_j |v_j\rangle\langle v_j|
\end{equation}
on $\operatorname{supp}(\rho)$, where $\{\lambda_j\}$ are the eigenvalues of $\Delta\Xi$ and $\{|v_j\rangle\}$ are the corresponding orthonormal
eigenvectors. Consequently,
\begin{equation}
    \operatorname{Tr}\!\left[\rho\, e^{\tau\Delta\Xi}\right]
    =
    \sum_j e^{\tau\lambda_j}\langle v_j|\rho|v_j\rangle.
\end{equation}
Because $\rho|_{\operatorname{supp}(\rho)}$ is strictly positive,
$\langle v_j|\rho|v_j\rangle>0$ for every eigenvector $|v_j\rangle$ of
$\Delta\Xi$. Hence
\(
\operatorname{Tr}\!\left[\rho\,e^{\tau\Delta\Xi}\right]
=
\sum_j
\langle v_j|\rho|v_j\rangle e^{\tau\lambda_j}
\)
is a finite positive linear combination of exponentials.

It follows that
$\tau\mapsto\operatorname{Tr}[\rho e^{\tau\Delta\Xi}]$
is entire and strictly positive on $\mathbb R$.
Hence $\mathcal F_{\rho,\sigma}(\tau)$
is real analytic on $\mathbb R$
and satisfies
\[
\mathcal F_{\rho,\sigma}(0)=0.
\]
The independence of $\Delta\Xi$ from $\tau$ implies
\begin{equation}
    \frac{d}{d\tau}e^{\tau\Delta\Xi}
    =
    \Delta\Xi e^{\tau\Delta\Xi}.
\end{equation}
Consequently, the derivatives of
$\mathcal F_{\rho,\sigma}(\tau)$
at $\tau=0$
generate the non-commutative cumulants of
$\Delta\Xi$.
By construction,
\begin{equation}
\label{eq:Sq_from_K}
    S_\alpha^Q(\rho\|\sigma)
    =
    \frac{\mathcal F_{\rho,\sigma}(\alpha-1)}
         {\alpha-1}.
\end{equation}

The analyticity of $\mathcal F_{\rho,\sigma}$ therefore provides a direct link between the Cu-Q relative Rényi functional and the cumulant
hierarchy of the quantum relative surprisal operator. Accordingly,
\begin{equation}
\label{eq:K_expansion}
    \mathcal F_{\rho,\sigma}(\tau)
    =
    \sum_{n=1}^{\infty}
    \frac{\tau^n}{n!}\,
    C_n(\Delta\Xi)_\rho,
\end{equation}
where
\begin{equation}
    C_n(\Delta\Xi)_\rho
    :=
    \left.
    \frac{d^n}{d\tau^n}
    \mathcal F_{\rho,\sigma}(\tau)
    \right|_{\tau=0}
\end{equation}
denotes the $n$th non-commutative cumulant of the quantum relative
surprisal operator. 

Substituting $\tau=\alpha-1$, which satisfies $\tau>0$ whenever $\alpha>1$, into Eq.\eqref{eq:K_expansion} and using Eq.\eqref{eq:Sq_from_K} yields
\begin{equation}
\label{eq:cumulant_expansion}
    S_\alpha^Q(\rho\|\sigma)
    =
    \sum_{n=1}^{\infty}
    \frac{(\alpha-1)^{n-1}}{n!}\,
    C_n(\Delta\Xi)_\rho.
\end{equation}
Thus, the Cu-Q relative Rényi framework naturally encodes the entire hierarchy of cumulants associated with the quantum relative surprisal
operator. Writing the first few terms explicitly,
\begin{align}
\label{eq:first_terms}
    S_\alpha^Q(\rho\|\sigma)
    &=
    C_1(\Delta\Xi)_\rho
    + \frac{\alpha-1}{2}\,C_2(\Delta\Xi)_\rho
    + \mathcal{O}\!\left((\alpha-1)^2\right) \nonumber \\
    &=
    S(\rho\|\sigma)
    + \frac{\alpha-1}{2}\operatorname{Var}_\rho(\Delta\Xi)
    + \mathcal{O}\!\left((\alpha-1)^2\right).
\end{align}
The first cumulant is
\begin{equation}
    C_1(\Delta\Xi)_\rho
    =
    \operatorname{Tr}\!\left[\rho(\ln\rho - \ln\sigma)\right],
\end{equation}
which coincides with the Umegaki quantum relative entropy $S(\rho\|\sigma)$.
The second cumulant is
\begin{equation}
    C_2(\Delta\Xi)_\rho
    =
    \operatorname{Tr}\!\left[\rho\,(\Delta\Xi)^2\right]
    -
    \Bigl(\operatorname{Tr}\!\left[\rho\,\Delta\Xi\right]\Bigr)^2,
\end{equation}
which is precisely the quantum relative variance
$\operatorname{Var}_\rho(\Delta\Xi)$. Higher-order cumulants ($C_n$, $n\ge 3$) encode increasingly refined statistical information, including non-commutative analogues of higher-order fluctuations. Further details regarding the cumulant-generating-function construction and its connection with the classical Rényi divergence are provided in Appendix\ref{app:derivation_cuq}.

\section{Cumulant path-integral-like representation}

In the classical setting, cumulant-generating functions encode fluctuation structure through their logarithmic moment-generating form, and their connection to large-deviation theory makes them well suited for information-theoretic applications~\cite{Touchette2009,Esposito2009}. In the quantum regime, non-commutativity of $\ln\rho$ and $\ln\sigma$ obstructs direct factorization, and ordering effects become intrinsic 
to the functional. We show that this operator structure induces a natural path-integral-like representation of the Cu-Q relative Rényi functional.

\begin{theorem}[Path-integral-like representation]
\label{prop:path_integral}
Let $\alpha>1$ and let $(\rho,\sigma)\in
\mathcal D(\mathcal H)\times
\mathcal D(\mathcal H)$
be density operators satisfying 
$\operatorname{supp}(\rho)\subseteq\operatorname{supp}(\sigma)$. 
Let $\{|\varphi_i\rangle\}$ and $\{|\psi_m\rangle\}$ denote the eigenbases of $\rho$ and $\sigma$, respectively, and define the transition probability 
$\mathbb{P}(\varphi_i|\psi_m):=|\langle\varphi_i|\psi_m\rangle|^2$. Then,
\begin{equation}
\label{eq:path_integral_form_new}
\begin{split}
S_\alpha^{Q}(\rho\|\sigma) 
= \frac{1}{\alpha-1}\ln\Biggl[ 
&\lim_{n\to\infty}
\sum_{\substack{\{(i_\beta,m_\beta)\}_{\beta=1}^n \\ i_n = i_0}} 
p(\varphi_{i_0}) \\
&\times \exp\Biggl(\sum_{\beta=1}^n 
\mathcal{A}_{\rho,\sigma}^{(\beta)}\Biggr)\Biggr],
\end{split}
\end{equation}
where the effective action per step is defined as
\begin{equation}
\begin{split}
\mathcal{A}_{\rho,\sigma}^{(\beta)} 
:= \,&\frac{1}{2}\ln\!\left( 
\mathbb{P}(\varphi_{i_\beta}|\psi_{m_\beta})
\mathbb{P}(\psi_{m_\beta}|\varphi_{i_{\beta-1}})
\right) \\
&+ \frac{\alpha-1}{n} \ln\frac{p(\varphi_{i_\beta})}{q(\psi_{m_\beta})} 
+ i\,\Delta\Theta_\beta,
\end{split}
\end{equation}
and
\begin{equation}
\begin{split}
\Delta\Theta_\beta 
&:= \Theta(\varphi_{i_\beta},\psi_{m_\beta}) 
   - \Theta(\varphi_{i_{\beta-1}},\psi_{m_\beta}), \\
\langle\varphi_i|\psi_m\rangle 
&= \sqrt{\mathbb{P}(\varphi_i|\psi_m)}\,e^{i\Theta(\varphi_i,\psi_m)}.
\end{split}
\end{equation}
\end{theorem}

\noindent 
\begin{proof}
We begin from the definition of Cu-Q relative Rényi functional in \eqref{QCRRF}. 
Define the generating functional
\begin{equation}\label{Geneq}
\Phi_{\rho,\sigma}(\alpha-1):=\operatorname{Tr}\!\left[ \rho\,e^{(\alpha-1)(\ln\rho-\ln\sigma)}\right].
\end{equation}
Since $\operatorname{supp}(\rho)\subseteq
\operatorname{supp}(\sigma)$, the operator
$\ln\rho-\ln\sigma$ is well defined on
$\operatorname{supp}(\rho)$. Since the Hilbert space is finite-dimensional, its restriction to
$\operatorname{supp}(\rho)$ is bounded and Hermitian. Applying the spectral decomposition of $\rho$, we obtain
\begin{equation}
\Phi_{\rho,\sigma}(\alpha-1)=\sum_{i_0} p(\varphi_{i_0}) \langle\varphi_{i_0}| e^{(\alpha-1)(\ln\rho-\ln\sigma)}
|\varphi_{i_0}\rangle.
\end{equation}
Using the Lie-Trotter product formula (in the strong operator topology),
\begin{equation}
e^{(\alpha-1)(\ln\rho-\ln\sigma)}=\lim_{n\to\infty}\left(
e^{\frac{\alpha-1}{n}\ln\rho}e^{-\frac{\alpha-1}{n}\ln\sigma} \right)^n,
\end{equation}
we write
\begin{equation}
\langle\varphi_{i_0}|e^{(\alpha-1)(\ln\rho-\ln\sigma)} |\varphi_{i_0}\rangle =\lim_{n\to\infty}\langle\varphi_{i_0}|
(\mathcal{T}_\rho \mathcal{T}_\sigma)^n|\varphi_{i_0}\rangle,
\end{equation}
where $\mathcal{T}_\rho := e^{\frac{\alpha-1}{n}\ln\rho},
\qquad\mathcal{T}_\sigma := e^{-\frac{\alpha-1}{n}\ln\sigma}$.

Expanding the product and inserting resolutions of the identity, we obtain
\begin{equation}
\langle\varphi_{i_0}|(\mathcal{T}_\rho \mathcal{T}_\sigma)^n|\varphi_{i_0}\rangle
=
\sum_{i_1,\dots,i_{n-1}}
\prod_{\beta=1}^n
\langle\varphi_{i_\beta}|
\mathcal{T}_\rho\mathcal{T}_\sigma
|\varphi_{i_{\beta-1}}\rangle,
\end{equation}
with \(i_n=i_0\), corresponding to cyclic paths induced by the trace operation. Inserting the resolution of the identity in the eigenbasis of \(\sigma\), we obtain
\begin{equation}
\begin{aligned}
&\langle\varphi_{i_\beta}| \mathcal{T}_\rho \mathcal{T}_\sigma |\varphi_{i_{\beta-1}}\rangle \\
&= \sum_{m_\beta} \langle\varphi_{i_\beta}|\psi_{m_\beta}\rangle \langle\psi_{m_\beta}|\varphi_{i_{\beta-1}}\rangle \exp\Biggl[ \frac{\alpha-1}{n} \ln \frac{p(\varphi_{i_\beta})}{q(\psi_{m_\beta})} \Biggr].
\end{aligned}
\end{equation}
Thus,
\begin{equation}
\Phi_{\rho,\sigma}(\alpha-1)
=
\lim_{n\to\infty}
\sum_{\substack{\{(i_\beta,m_\beta)\} \\ i_n=i_0}}
p(\varphi_{i_0})
\prod_{\beta=1}^n
\mathcal{M}_\beta,
\end{equation}
where
\begin{equation}
\mathcal{M}_\beta=\langle\varphi_{i_\beta}|\psi_{m_\beta}\rangle
\langle\psi_{m_\beta}|\varphi_{i_{\beta-1}}\rangle\exp\!\left(
\frac{\alpha-1}{n}\ln\frac{p(\varphi_{i_\beta})}{q(\psi_{m_\beta})}
\right).
\end{equation}
Using the polar decomposition,
\begin{equation}
\langle\varphi_i|\psi_m\rangle = \sqrt{\mathbb{P}(\varphi_i|\psi_m)}\,e^{i\Theta(\varphi_i,\psi_m)},
\end{equation}
we obtain
\begin{equation}
\begin{aligned}
\mathcal{M}_\beta 
&= \exp\Biggl( \frac{1}{2}\ln\left( \mathbb{P}(\varphi_{i_\beta}|\psi_{m_\beta})\mathbb{P}(\psi_{m_\beta}|\varphi_{i_{\beta-1}})\right) \\
&\quad + \frac{\alpha-1}{n} \ln\frac{p(\varphi_{i_\beta})}{q(\psi_{m_\beta})} + i\,\Delta\Theta_\beta \Biggr).
\end{aligned}
\end{equation}
Taking the product,
\begin{equation}
\prod_{\beta=1}^n \mathcal{M}_\beta =\exp\!\left(\sum_{\beta=1}^n
\mathcal{A}_{\rho,\sigma}^{(\beta)}\right).
\end{equation}
Substituting back yields
\begin{equation}
\Phi_{\rho,\sigma}(\alpha-1)=\lim_{n\to\infty}\sum_{\substack{\{(i_\beta,m_\beta)\} \\ i_n=i_0}}p(\varphi_{i_0})\exp\!\left(
\sum_{\beta=1}^n\mathcal{A}_{\rho,\sigma}^{(\beta)}\right),
\end{equation}
which proves \eqref{eq:path_integral_form_new}.
\end{proof}

The representation in Theorem~\ref{prop:path_integral} 
admits a natural interpretation as a weighted sum over discrete cyclic trajectories in Hilbert space. Each trajectory is specified by a sequence 
$\{(i_\beta, m_\beta)\}_{\beta=1}^n$ of transitions 
between the eigenbases of $\rho$ and $\sigma$, and 
contributes a complex weight of the form
\[
\exp\left(\sum_{\beta=1}^n \mathcal{A}_{\rho,\sigma}^{(\beta)}\right).
\]
This representation may be viewed as a formal complex-valued analogue of a partition function over a path ensemble. The real part of the exponent,
\[
\frac{1}{2}\ln\!\left(
\mathbb{P}(\varphi_{i_\beta}|\psi_{m_\beta})
\mathbb{P}(\psi_{m_\beta}|\varphi_{i_{\beta-1}})
\right)
+
\frac{\alpha-1}{n}
\ln\frac{p(\varphi_{i_\beta})}{q(\psi_{m_\beta})},
\]
encodes the statistical weight associated with transitions between eigenstates and the relative scaling of eigenvalues. In contrast, the imaginary part, given by the phase increments $\Delta\Theta_\beta$, captures interference 
effects arising from the non-commutativity between 
$\rho$ and $\sigma$.

Although the representation may be rewritten in a form reminiscent of a Feynman path integral, its interpretation here is fundamentally different. Instead, it emerges from the polar decomposition of 
the overlaps between the eigenbases of $\rho$ and 
$\sigma$, reflecting the geometric phase accumulated along each path due to the non-commutativity of the 
two operators.

In the commuting case, where $[\rho,\sigma]=0$, the eigenbases coincide, and the overlaps reduce to Kronecker delta functions. Consequently, the path structure degenerates entirely: the phase increments vanish identically, and the sum reduces directly to a classical weighted sum, recovering the classical Rényi divergence as in Proposition~\ref{prop:classical-reduction}.

We emphasize that this representation is purely algebraic and does not define a probability measure over paths, but rather provides a structural decomposition of the operator expression into complex-weighted contributions.

\section{{Analytical properties}\label{Analytical property}}

In this section, we investigate several analytical properties of the Cu-Q relative Rényi functional \(S_\alpha^Q(\rho\|\sigma)\).
In particular, we establish continuity, normalization, unitary invariance, additivity, compatibility with generalized mean structures, positivity, and monotonicity with respect to the Rényi parameter \(\alpha\).
These results demonstrate that \(S_\alpha^Q(\rho\|\sigma)\) is mathematically well defined on its studied regime and satisfies several properties commonly expected of Rényi-type quantum divergences.
The quantum data-processing inequality (QDPI) is discussed separately in the following section.

The majority of the analytical results in this section are established on the studied regime of the Cu-Q relative Rényi functional, namely
\(
\alpha>1
\)
together with
\(
\operatorname{supp}(\rho)\subseteq\operatorname{supp}(\sigma),
\)
where logarithms are understood on the corresponding supports.

For continuity, two complementary perspectives are considered. First, continuity is established on the strictly positive state space
\(
\mathcal D_+(\mathcal H)\times\mathcal D_+(\mathcal H),
\)
where the logarithm behaves regularly under functional calculus. Second, regularization is introduced to study the associated regularized functional and its dependence on the regularization parameter.

Regularized constructions also play an important role
in later sections, particularly in the investigation of the \(\alpha=0\) regime and in extending the parameter
domain of the regularized functional beyond that of the
original definition. In particular, this leads to a quantity that naturally characterizes the non-commutativity between quantum states and may be interpreted as a measure of relative quantumness.

\begin{lemma}[H\"{o}lder inequality for Schatten classes]\label{Lem1}
Let $A \in \mathcal{L}_p(\mathcal{H})$ and $B \in \mathcal{L}_q(\mathcal{H})$ denote 
the Schatten $p$- and $q$-class operators respectively, with $1/p + 1/q = 1$.
Then
\begin{equation}
    \big|\operatorname{Tr}(AB)\big| \le \|A\|_p \, \|B\|_q.
\end{equation}
In particular, for $p=1$ and $q=\infty$,
\begin{equation}
    \big|\operatorname{Tr}(AB)\big| \le \|A\|_1 \, \|B\|_\infty.
\end{equation}
\end{lemma}

We first establish continuity on the full-rank state space,
which is contained in the studied regime. A subsequent result studies the continuity of the regularized Cu-Q relative Rényi functional with respect to the regularization parameter.

\begin{proposition}[Continuity]\label{Prop:Continuity}
Let $\mathcal{H}$ be a finite-dimensional Hilbert space and let 
$\{(\rho_n,\sigma_n)\}_{n\in\mathbb{N}} \subset \mathcal{D}_+(\mathcal{H}) \times \mathcal{D}_+(\mathcal{H})$ 
be a sequence of full-rank density operators converging to $(\rho,\sigma) \in \mathcal{D}_+(\mathcal{H}) \times \mathcal{D}_+(\mathcal{H})$ in trace norm, $\|\cdot\|_1$. Then, for any fixed $\alpha>1$,
\begin{equation}
    \lim_{n\to\infty} S^Q_\alpha(\rho_n\|\sigma_n) = S^Q_\alpha(\rho\|\sigma),
\end{equation}
that is, the cumulant-based quantum relative Rényi functional $S^Q_\alpha(\rho\|\sigma)$ is continuous on 
$\mathcal{D}_+(\mathcal{H}) \times \mathcal{D}_+(\mathcal{H})$.
\end{proposition}

\begin{remark}
The full-rank assumption is imposed only for the continuity argument presented here.
It guarantees that the spectra of the states remain uniformly bounded away from zero, so that the logarithm is continuous under functional calculus and the operator \(\ln\rho-\ln\sigma\)
depends continuously on the pair \((\rho,\sigma).\)
\end{remark}

\begin{proof}
Let us consider the generating functional 
$\Phi_{\rho,\sigma}(\alpha-1)$ in~\eqref{Geneq}. In finite-dimensional Hilbert spaces, the Schatten norms satisfy $\|A\|_\infty \le \|A\|_1$ for any operator $A$. Therefore, convergence in trace norm implies convergence in operator norm, which yields $\|\rho_n-\rho\|_\infty \to 0$ and $\|\sigma_n-\sigma\|_\infty \to 0$. 

Since $\rho$ and $\sigma$ are of full rank, their minimal eigenvalues are strictly positive.
Moreover, since density operators have spectrum contained in $[0,1]$, and full-rankness guarantees a strictly positive lower spectral bound, there exists $c>0$ and $N_0 \in \mathbb{N}$ such that for all $n \ge N_0$,
\begin{equation}
\rho_n,\sigma_n,\rho,\sigma \ge cI.
\end{equation}
By continuity of eigenvalues with respect to the operator norm, such a uniform bound is stable under perturbation. 

Thus, their spectra lie in the compact interval $[c,1]$. 
Since $\log$ is continuous on $[c,1]$, functional calculus yields 
$\|\ln\rho_n-\ln\rho\|_\infty \to 0$ and $\|\ln\sigma_n-\ln\sigma\|_\infty \to 0$.

Set
\begin{equation}
    X_n := (\alpha-1)(\ln\rho_n-\ln\sigma_n)
\end{equation}
and 
\begin{equation}
    X := (\alpha-1)(\ln\rho-\ln\sigma).
\end{equation}
Using linearity of the operator norm and the triangle inequality, we estimate 
$\|X_n - X\|_\infty$ as follows. First, substituting the definitions of $X_n$ and 
$X$ and factoring out $|\alpha-1|$ via $\|cA\|_\infty = |c|\|A\|_\infty$,
\begin{align}
    \|X_n-X\|_{\infty} 
    &= \bigl\|(\alpha-1)(\ln\rho_n-\ln\sigma_n) \nonumber \\
    &\quad - (\alpha-1)(\ln\rho-\ln\sigma)\bigr\|_{\infty} \nonumber \\
    &= |\alpha-1|\,\Bigl\|(\ln\rho_n - \ln\sigma_n) \nonumber \\
    &\quad - (\ln\rho - \ln\sigma)\Bigr\|_{\infty}.
\end{align}
Regrouping the terms inside the norm as 
$(\ln\rho_n - \ln\rho) - (\ln\sigma_n - \ln\sigma)$
and applying the triangle inequality $\|A - B\|_\infty \le \|A\|_\infty + \|B\|_\infty$,
which follows from $\|A+(-B)\|_\infty \le \|A\|_\infty + \|-B\|_\infty$ and 
homogeneity $\|-B\|_\infty = \|B\|_\infty$, we obtain
\begin{align}
    \|X_n-X\|_{\infty} 
    &\le |\alpha-1|\,\bigl(\|\ln\rho_n - \ln\rho\|_{\infty} \notag\\
    &\qquad + \|\ln\sigma_n - \ln\sigma\|_{\infty}\bigr).
\end{align}
Since $\|\ln\rho_n - \ln\rho\|_\infty \to 0$ and $\|\ln\sigma_n - \ln\sigma\|_\infty \to 0$ 
by functional calculus, and $|\alpha-1|$ is a finite constant, we conclude
\begin{equation}
    \|X_n - X\|_\infty \to 0.
\end{equation}
The exponential map is continuous with respect to the operator norm on bounded 
operators, that is, $\|A_n - A\|_\infty \to 0$ implies 
$\|e^{A_n} - e^{A}\|_\infty \to 0$. Hence,
\begin{equation}
    \|e^{X_n}-e^{X}\|_{\infty}\to 0.
\end{equation}
Since $X_n\to X$ in operator norm, by definition of convergence 
applied with $\varepsilon = 1$, there exists $N_1\in\mathbb{N}$ such that
\begin{equation}
    \|X_n - X\|_{\infty} < 1 \qquad \text{for all } n \ge N_1.
\end{equation}
Applying the triangle inequality,
\begin{align}
    \|X_n\|_{\infty} 
    &= \|X_n - X + X\|_{\infty} \notag\\ 
    &\le \underbrace{\|X_n - X\|_{\infty}}_{<\, 1} + \|X\|_{\infty} \notag\\
    &< \|X\|_{\infty} + 1
    \qquad \text{for all } n\ge N_1.
\end{align}
Set $r := \|X\|_{\infty}+1 > 0$, so that $\|X_n\|_{\infty} < r$ for all $n \ge N_1$.
To bound $\|e^{X_n}\|_\infty$, we use the power series expansion
\begin{equation}
    e^{X_n} = \sum_{k=0}^{\infty} \frac{X_n^k}{k!} 
\end{equation}
Taking the operator norm of both sides and applying the triangle inequality for 
infinite sums $\bigl\|\sum_{k=0}^{\infty} A_k\bigr\|_\infty \le \sum_{k=0}^{\infty} 
\|A_k\|_\infty$, we obtain
\begin{equation}
    \|e^{X_n}\|_{\infty} 
    = \left\|\sum_{k=0}^{\infty} \frac{X_n^k}{k!}\right\|_\infty
    \le \sum_{k=0}^{\infty} \frac{\|X_n^k\|_{\infty}}{k!}.
\end{equation}
Applying submultiplicativity of the operator norm, that is, 
$\|AB\|_\infty \le \|A\|_\infty\|B\|_\infty$, iteratively gives
\begin{equation}
\begin{aligned}
    \|X_n^k\|_\infty 
    &= \|X_n \cdot X_n \cdots X_n\|_\infty \\
    &\le \|X_n\|_\infty \cdot \|X_n\|_\infty \cdots \|X_n\|_\infty \\
    &= \|X_n\|_\infty^k.
\end{aligned}
\end{equation}
Substituting this bound,
\begin{equation}
    \sum_{k=0}^{\infty} \frac{\|X_n^k\|_{\infty}}{k!} 
    \le \sum_{k=0}^{\infty} \frac{\|X_n\|_{\infty}^k}{k!}.
\end{equation}
Since $\|X_n\|_\infty$ is a real number, the right-hand side is exactly the 
power series of the real exponential, so
\begin{equation}
    \sum_{k=0}^{\infty} \frac{\|X_n\|_{\infty}^k}{k!} = e^{\|X_n\|_{\infty}}.
\end{equation}
Combining the above,
\begin{equation}
    \|e^{X_n}\|_{\infty} \le \sum_{k=0}^{\infty} \frac{\|X_n^k\|_\infty}{k!}
    \le \sum_{k=0}^{\infty} \frac{\|X_n\|_{\infty}^k}{k!} 
    = e^{\|X_n\|_{\infty}}.
\end{equation}
Because $e^{(\cdot)}$ is increasing on $\mathbb{R}$ and $\|X_n\|_\infty < r$ 
for all $n \ge N_1$, we conclude
\begin{equation}
    \|e^{X_n}\|_{\infty} \le e^{\|X_n\|_{\infty}} \le e^{r}
    \qquad \text{for all } n\ge N_1.
\end{equation}

To establish convergence of $\Phi_{\rho_n,\sigma_n}(\alpha-1)$,
consider the difference
\begin{equation}
    \Delta_n =
    \big|
    \operatorname{Tr}(\rho_n e^{X_n})
    -
    \operatorname{Tr}(\rho e^{X})
    \big|.
\end{equation}
By adding and subtracting $\operatorname{Tr}(\rho e^{X_n})$ and applying the
triangle inequality,
\begin{equation}
    \Delta_n \le \big|\operatorname{Tr}((\rho_n-\rho)e^{X_n})\big|+
    \big|\operatorname{Tr}(\rho(e^{X_n}-e^{X}))\big|.
\end{equation}
From Lemma~\ref{Lem1}, we obtain 
\begin{equation}
    \big|\operatorname{Tr}((\rho_n-\rho)e^{X_n})\big|\le\|\rho_n-\rho\|_1 \, \|e^{X_n}\|_\infty,
\end{equation}
and, using $\|\rho\|_1 = \operatorname{Tr}(\rho) = 1$,
\begin{equation}
    \big|\operatorname{Tr}(\rho(e^{X_n}-e^{X}))\big|\le
    \|\rho\|_1 \, \|e^{X_n}-e^{X}\|_{\infty} = \|e^{X_n}-e^{X}\|_{\infty}.
\end{equation}
Thus, 
\begin{equation}
    \Delta_n\le \|\rho_n-\rho\|_1 \, \|e^{X_n}\|_{\infty}
    + \|e^{X_n}-e^{X}\|_{\infty}.
\end{equation}
Since $\|\rho_n-\rho\|_1 \to 0$, for any $\varepsilon_1>0$ there exists $N_2\in\mathbb{N}$ 
such that
\begin{equation}
    \|\rho_n-\rho\|_1 < \frac{\varepsilon_1}{e^{r}}
    \quad \text{for all } n\ge N_2.
\end{equation}
Consequently, for all $n\ge \max\{N_1,N_2\}$,
\begin{equation}
    \|\rho_n-\rho\|_1\,\|e^{X_n}\|_{\infty} < \varepsilon_1.
\end{equation}    

On the other hand, the continuity of the exponential map in operator norm implies
$\|e^{X_n}-e^{X}\|_{\infty}\to 0$. Hence, for any $\varepsilon_2>0$ there exists 
$N_3\in\mathbb{N}$ such that
\begin{equation}
    \|e^{X_n}-e^{X}\|_{\infty} < \varepsilon_2
    \quad \text{for all } n\ge N_3.
\end{equation}
Letting $N := \max\{N_0, N_1, N_2, N_3\}$, we obtain for all $n\ge N$,
\begin{equation}
    \Delta_n \le \|\rho_n-\rho\|_1\,\|e^{X_n}\|_{\infty}+
    \|e^{X_n}-e^{X}\|_{\infty}<\varepsilon_1+\varepsilon_2.
\end{equation}
Since $\varepsilon_1$ and $\varepsilon_2$ are arbitrary with
$\varepsilon_1+\varepsilon_2=\varepsilon$, it follows that
$\Delta_n\to 0$ as $n\to\infty$. Hence,
\[
    \Phi_{\rho_n,\sigma_n}(\alpha-1)
    \longrightarrow
    \Phi_{\rho,\sigma}(\alpha-1).
\]
Since $e^{X_n}\succ 0$ for all $n$ and $e^X\succ 0$, and both $\rho_n$ and $\rho$
are density operators, we have
\begin{equation}
    \Phi_{\rho_n,\sigma_n}(\alpha-1)=\operatorname{Tr}(\rho_n e^{X_n})>0
\end{equation}
and
\begin{equation}
     \Phi_{\rho,\sigma}(\alpha-1)=\operatorname{Tr}(\rho e^{X})>0
\end{equation}
for all $n\in\mathbb{N}$, so the logarithm is well defined along the entire sequence and at the limit point. Therefore, by continuity of the logarithm on
\((0,\infty)\),
\[
\log\Phi_{\rho_n,\sigma_n}(\alpha-1)
\longrightarrow
\log\Phi_{\rho,\sigma}(\alpha-1).
\]
Consequently,
\[
S_\alpha^Q(\rho_n\|\sigma_n)
\longrightarrow
S_\alpha^Q(\rho\|\sigma).
\]
This completes the proof of continuity of $S^Q_\alpha(\rho\|\sigma)$
on $\mathcal{D}_+(\mathcal{H}) \times \mathcal{D}_+(\mathcal{H})$.
\end{proof}

From a mathematical standpoint, Proposition~\ref{Prop:Continuity}
shows that \(S_\alpha^Q(\rho\|\sigma)\) is continuous on
\(
\mathcal D_+(\mathcal H)\times\mathcal D_+(\mathcal H)
\)
with respect to the trace-norm topology. From a physical standpoint, continuity reflects the stability of the distinguishability measure under small perturbations of the states.

Many quantum states encountered in physical applications, however, are not necessarily of full rank. The regularization procedure introduced in Eqs.~\eqref{state_regularized}
and~\eqref{regularized_Cu-Q}
provides a one-parameter family of strictly positive
approximations of such states.

Since the regularized functional is defined through
this family of full-rank states, it is natural to ask whether its value depends
continuously on the regularization parameter
\(\epsilon\). The following corollary shows that this is indeed the case.

\begin{corollary}[Continuity of the regularized functional]
\label{Prop:RegularizedContinuity}
Let $(\rho,\sigma)\in \mathcal{D}(\mathcal{H})\times \mathcal{D}(\mathcal{H})$
and let $\rho_\epsilon$ and $\sigma_\epsilon$ be the regularized states defined
in Eq.~\eqref{state_regularized}.
For every fixed $\alpha > 1$, the mapping
\[
  \epsilon \longmapsto S_\alpha^{\mathcal{Q}}(\rho_\epsilon\|\sigma_\epsilon)
\]
is continuous on $(0,1]$.
\end{corollary}

\begin{proof}
Since ${\mathbb{I}}/d \succ 0$, both $\rho_\epsilon$ and $\sigma_\epsilon$ are full-rank.
Since $\rho_\epsilon$ and $\sigma_\epsilon$ depend affinely on $\epsilon$,
\[
  \|\rho_\epsilon - \rho_{\epsilon_0}\|_1
  = |\epsilon - \epsilon_0|\,\Bigl\|\rho - \frac{{\mathbb{I}}}{d}\Bigr\|_1 \to 0,
\]
and similarly $\|\sigma_\epsilon - \sigma_{\epsilon_0}\|_1 \to 0$ as
$\epsilon \to \epsilon_0$. Thus
$(\rho_\epsilon, \sigma_\epsilon) \to (\rho_{\epsilon_0}, \sigma_{\epsilon_0})$
in trace norm, with the limit pair full-rank.
Applying Proposition~\ref{Prop:Continuity} yields
\[
  \lim_{\epsilon \to \epsilon_0}
  S_\alpha^{\mathcal{Q}}(\rho_\epsilon \| \sigma_\epsilon)
  = S_\alpha^{\mathcal{Q}}(\rho_{\epsilon_0} \| \sigma_{\epsilon_0}),
\]
which completes the proof.
\end{proof}

\begin{remark}
The regularized functional $S_\alpha^{\mathcal{Q}}(\rho_\epsilon\|\sigma_\epsilon)$
is defined for all $\alpha \in \mathbb{R} \setminus \{1\}$, since the full-rank
condition on $\rho_\epsilon$ and $\sigma_\epsilon$ ensures that the logarithms
are well defined on the entire Hilbert space. The continuity argument above
extends to this broader parameter domain, as the proof relies only on the
full-rank property of the regularized states and the finiteness of $|\alpha-1|$,
neither of which depends on the sign of $\alpha - 1$. Consequently, the mapping
$\epsilon \mapsto S_\alpha^{\mathcal{Q}}(\rho_\epsilon\|\sigma_\epsilon)$
is continuous on $(0,1]$ for all $\alpha \in \mathbb{R} \setminus \{1\}$.
\end{remark}

The continuity results above establish the stability of the Cu-Q relative Rényi functional under perturbations of the states. Another fundamental requirement for any divergence measure is normalization, namely that the divergence vanishes when the two states
coincide.
    
\begin{proposition}[Normalization]\label{Prop:Normalization}
Let $\alpha > 1$ and let $(\rho,\sigma)\in\mathcal{D}(\mathcal{H})
\times\mathcal{D}(\mathcal{H})$ be density operators satisfying 
$\operatorname{supp}(\rho)\subseteq\operatorname{supp}(\sigma)$. 
Then $S^Q_\alpha(\rho\|\rho) = 0$.
\end{proposition}

\begin{proof}
Setting $\sigma = \rho$, we have
\begin{equation}
    (\alpha-1)(\ln\rho-\ln\rho) = 0,
\end{equation}
and hence
\begin{equation}
    e^{(\alpha-1)(\ln\rho-\ln\rho)} = \mathbb{I}.
\end{equation}
Substituting into~\eqref{Geneq}, we obtain
\begin{equation}
    \Phi_{\rho,\rho}(\alpha-1) =\operatorname{Tr}\!\left(\rho\,\mathbb{I}\right)
    =\operatorname{Tr}(\rho)=1,
\end{equation}
where the last equality follows from the unit-trace property of density operators.
Therefore,
\begin{equation}
    S^Q_\alpha(\rho\|\rho) = \frac{1}{\alpha-1}\ln 1 = 0.
\end{equation}
\end{proof}

\begin{example}[Scaling behavior outside the normalized setting]
\label{ex:scaling}
To illustrate the scaling behavior of the generating functional, consider the positive operators $\rho=\mathbb{I}$ and $\sigma=\mathbb{I}/2$ on a $d$-dimensional Hilbert space. These operators do not satisfy the unit-trace condition and hence do not belong to $\mathcal{D}(\mathcal{H})$, but they provide a useful algebraic illustration.

Since both operators are scalar multiples of $\mathbb{I}$, one has
\begin{equation}
    \ln\rho - \ln\sigma
    =
    \ln\mathbb{I} - \ln(\mathbb{I}/2)
    =
    \ln 2\,\mathbb{I},
\end{equation}
and therefore
\begin{equation}
    \Phi_{\rho,\sigma}(\alpha-1)
    =
    \operatorname{Tr}\!\left[
    \mathbb{I}\,e^{(\alpha-1)\ln 2\,\mathbb{I}}
    \right]
    =
    \operatorname{Tr}\!\left[2^{\alpha-1}\,\mathbb{I}\right]
    =
    2^{\alpha-1} d.
\end{equation}
It follows that
\begin{equation}
    S^Q_\alpha(\mathbb{I}\|\tfrac{\mathbb{I}}{2})
    =
    \frac{1}{\alpha-1}\ln\bigl(2^{\alpha-1} d\bigr)
    =
    \ln 2 + \frac{\ln d}{\alpha-1}.
\end{equation}
The additional term $\frac{\ln d}{\alpha-1}$ reflects the absence of 
unit-trace normalization. Upon rescaling to density operators 
$\rho'=\mathbb{I}/d$ and $\sigma'=\mathbb{I}/(2d)$, one obtains
\begin{equation}
    S^Q_\alpha(\rho'\|\sigma')
    =
    \ln 2,
\end{equation}
which coincides with the classical Rényi divergence in this commuting setting.
\end{example}

\medskip

Proposition~\ref{Prop:Normalization} establishes that the functional vanishes when both arguments coincide, as expected for a divergence measure. Example~\ref{ex:scaling} further illustrates that, in a commuting setting, the functional reduces to a value determined solely by the relative scaling between the operators, once proper normalization is imposed.

Continuity and normalization establish the basic stability and scale of $S^Q_\alpha(\rho\|\sigma)$, but say nothing about its behaviour under changes of representation. We therefore next examine whether $S^Q_\alpha(\rho\|\sigma)$ is invariant under unitary transformations.

\begin{proposition}[Unitary invariance]
Let $\alpha >1$ and let
$(\rho,\sigma) \in \mathcal{D}(\mathcal{H}) \times \mathcal{D}(\mathcal{H})$
be density operators satisfying 
$\operatorname{supp}(\rho)\subseteq\operatorname{supp}(\sigma)$.
Then, for any unitary operator $U$ on $\mathcal{H}$,
\begin{equation}
S^Q_\alpha(U\rho U^\dagger \| U\sigma U^\dagger)
=
S^Q_\alpha(\rho \| \sigma).
\end{equation}
\end{proposition}

\begin{proof}
By definition of $S^Q_\alpha(\rho \| \sigma)$, we have
\begin{equation}
\begin{aligned}
&S^Q_\alpha(U\rho U^\dagger \| U\sigma U^\dagger) \\
&= \frac{1}{\alpha-1} \ln \operatorname{Tr} \left[ (U\rho U^\dagger) e^{(\alpha-1)\bigl(\ln(U\rho U^\dagger)-\ln(U\sigma U^\dagger)\bigr)} \right].
\end{aligned}
\end{equation}
By the unitary covariance of the functional calculus, for any positive operator $G$
and any continuous function $f$, we have
\begin{equation}
f(UGU^\dagger)
=
Uf(G)U^\dagger.
\end{equation}
In particular, for $f(x)=\ln x$,
\begin{equation}
\ln(U G U^\dagger)
=
U(\ln G)U^\dagger.
\end{equation}
\begin{equation}
\begin{aligned}
    \ln(U\rho U^\dagger) - \ln(U\sigma U^\dagger) 
    &= U(\ln\rho)U^\dagger - U(\ln\sigma)U^\dagger \\
    &= U(\ln\rho - \ln\sigma)U^\dagger.
\end{aligned}
\end{equation}
Again using the unitary covariance of the functional calculus, now for the 
exponential map,
\begin{equation}
e^{UAU^\dagger}
=
U e^A U^\dagger,
\end{equation}
we obtain
\begin{equation}
e^{(\alpha-1)U(\ln\rho-\ln\sigma)U^\dagger}
=
U e^{(\alpha-1)(\ln\rho-\ln\sigma)} U^\dagger.
\end{equation}
Substituting back gives
\begin{equation}
\begin{aligned}
&S^Q_\alpha(U\rho U^\dagger \| U\sigma U^\dagger) \\
&= \frac{1}{\alpha-1} \ln \operatorname{Tr} \left[ U\rho U^\dagger \cdot U e^{(\alpha-1)(\ln\rho-\ln\sigma)} U^\dagger \right].
\end{aligned}
\end{equation}
Since $U^\dagger U = \mathbb{I}$, this simplifies to
\begin{equation}
\begin{aligned}
&S^Q_\alpha(U\rho U^\dagger \| U\sigma U^\dagger) \\
&= \frac{1}{\alpha-1} \ln \operatorname{Tr} \left[ U\rho e^{(\alpha-1)(\ln\rho-\ln\sigma)} U^\dagger \right].
\end{aligned}
\end{equation}
Finally, using the cyclicity of the trace, $\operatorname{Tr}(UAU^\dagger) = 
\operatorname{Tr}(A)$, we obtain
\begin{equation}
\operatorname{Tr}
\left[
U\rho e^{(\alpha-1)(\ln\rho-\ln\sigma)} U^\dagger
\right]
=
\operatorname{Tr}
\left[
\rho e^{(\alpha-1)(\ln\rho-\ln\sigma)}
\right].
\end{equation}
Therefore,
\begin{equation}
S^Q_\alpha(U\rho U^\dagger \| U\sigma U^\dagger)
=
S^Q_\alpha(\rho \| \sigma),
\end{equation}
which completes the proof.
\end{proof}

This proof demonstrates that $S^Q_\alpha(\rho \| \sigma)$ depends solely on the physical content of the states, rather than the choice of basis. Since unitary transformations correspond to changes of basis or reference frames, two observers describing the same system in different bases will necessarily agree on the value of the divergence. This property is the quantum analogue of the classical fact that any reasonable divergence is invariant under permutations of outcomes.

Furthermore, a meaningful functional should behave predictably under the composition of independent systems. Accordingly, we next investigate the compositional properties of $S^Q_\alpha(\rho\|\sigma)$, beginning with additivity and its extension to generalized mean structures.

\begin{proposition}[Additivity under tensor products]
Let $\alpha >1$, let $\mathcal{H}_1$ and $\mathcal{H}_2$ 
be finite-dimensional Hilbert spaces, and let
$(\rho,\sigma) \in \mathcal{D}(\mathcal{H}_1) \times \mathcal{D}(\mathcal{H}_1)$ 
and $(\eta,\Omega) \in \mathcal{D}(\mathcal{H}_2) \times \mathcal{D}(\mathcal{H}_2)$
be density operators satisfying
$\operatorname{supp}(\rho)\subseteq\operatorname{supp}(\sigma)$ and
$\operatorname{supp}(\eta)\subseteq\operatorname{supp}(\Omega)$.
Then,
\begin{equation}
S^Q_{\alpha}(\rho \otimes \eta \| \sigma \otimes \Omega)
=
S^Q_{\alpha}(\rho \| \sigma)
+
S^Q_{\alpha}(\eta \| \Omega).
\label{eq:additivity}
\end{equation}
\end{proposition}

\begin{proof}
Applying the definition of $S^Q_\alpha(\rho \| \sigma)$ to the tensor-product states, we consider
\begin{equation}
\operatorname{Tr}\!\left[
(\rho \otimes \eta)
e^{(\alpha-1)\bigl(\ln(\rho \otimes \eta)-\ln(\sigma \otimes \Omega)\bigr)}
\right].
\end{equation}
Since $\rho \otimes \eta$ is a positive operator on 
$\mathcal{H}_1 \otimes \mathcal{H}_2$, functional calculus restricted 
to the support gives
\begin{equation}
\ln(\rho \otimes \eta)
=
\ln\rho \otimes \mathbb{I}
+
\mathbb{I} \otimes \ln\eta,
\end{equation}
and similarly,
\begin{equation}
\ln(\sigma \otimes \Omega)
=
\ln\sigma \otimes \mathbb{I}
+
\mathbb{I} \otimes \ln\Omega.
\end{equation}
Subtracting,
\begin{equation}
\ln(\rho \otimes \eta)
-
\ln(\sigma \otimes \Omega)
=
(\ln\rho-\ln\sigma)\otimes \mathbb{I}
+
\mathbb{I}\otimes(\ln\eta-\ln\Omega).
\end{equation}
Since $A \otimes \mathbb{I}$ and $\mathbb{I} \otimes B$ commute for any 
operators $A$ and $B$, the exponential of the sum splits as a tensor product,
\begin{equation}
\begin{aligned}
&\exp\Bigl[ (\alpha-1)\bigl( (\ln\rho-\ln\sigma)\otimes \mathbb{I} + \mathbb{I}\otimes(\ln\eta-\ln\Omega) \bigr) \Bigr] \\
&= e^{(\alpha-1)(\ln\rho-\ln\sigma)} \otimes e^{(\alpha-1)(\ln\eta-\ln\Omega)}.
\end{aligned}
\end{equation}
Using the mixed-product property 
$(A \otimes B)(C \otimes D) = (AC) \otimes (BD)$,
we simplify
\begin{equation}
\begin{aligned}
&(\rho \otimes \eta) \bigl( e^{(\alpha-1)(\ln\rho-\ln\sigma)} \otimes e^{(\alpha-1)(\ln\eta-\ln\Omega)} \bigr) \\
&= \rho\,e^{(\alpha-1)(\ln\rho-\ln\sigma)} \otimes \eta\,e^{(\alpha-1)(\ln\eta-\ln\Omega)}.
\end{aligned}
\end{equation}
Applying multiplicativity of the trace, 
$\operatorname{Tr}(A \otimes B) = \operatorname{Tr}(A)\operatorname{Tr}(B)$,
we obtain
\begin{equation}
\begin{aligned}
&\operatorname{Tr}\Bigl[ (\rho \otimes \eta) \bigl( e^{(\alpha-1)(\ln\rho-\ln\sigma)} \otimes e^{(\alpha-1)(\ln\eta-\ln\Omega)} \bigr) \Bigr] \\
&= \operatorname{Tr}\bigl[ \rho\,e^{(\alpha-1)(\ln\rho-\ln\sigma)} \bigr] \operatorname{Tr}\bigl[ \eta\,e^{(\alpha-1)(\ln\eta-\ln\Omega)} \bigr].
\end{aligned}
\end{equation}
Therefore,
\begin{equation}
\begin{aligned}
&S^Q_{\alpha}(\rho \otimes \eta \| \sigma \otimes \Omega) \\
&= \frac{1}{\alpha-1} \ln \Biggl( \operatorname{Tr}\bigl[ \rho\,e^{(\alpha-1)(\ln\rho-\ln\sigma)} \bigr] \\
&\quad \times \operatorname{Tr}\bigl[ \eta\,e^{(\alpha-1)(\ln\eta-\ln\Omega)} \bigr] \Biggr).
\end{aligned}
\end{equation}
and using $\ln(AB)=\ln A+\ln B$, we conclude
\begin{equation}
S^Q_{\alpha}(\rho \otimes \eta \| \sigma \otimes \Omega)
=
S^Q_{\alpha}(\rho \| \sigma)
+
S^Q_{\alpha}(\eta \| \Omega).
\end{equation}
\end{proof}

Additivity shows that $S^Q_\alpha(\rho \| \sigma)$ accumulates information additively across 
independent systems: the functional between two joint states $\rho \otimes \eta$ and $\sigma \otimes \Omega$ is exactly the sum of the functionals of the individual components. 

A related but structurally different question is how $S^Q_\alpha(\rho \| \sigma)$  behaves under direct sums, where the combination is no longer multiplicative but additive at the level of the Hilbert space. This leads to a generalized mean structure, as we now show.

\begin{proposition}[Generalized mean under direct sums]
Let $\alpha >1$, let $\mathcal{H}_1$ and $\mathcal{H}_2$ 
be finite-dimensional Hilbert spaces, and let
$(\rho,\sigma) \in \mathcal{D}(\mathcal{H}_1) \times \mathcal{D}(\mathcal{H}_1)$ 
and $(\eta,\Omega) \in \mathcal{D}(\mathcal{H}_2) \times \mathcal{D}(\mathcal{H}_2)$
be density operators satisfying
$\operatorname{supp}(\rho)\subseteq\operatorname{supp}(\sigma)$ and
$\operatorname{supp}(\eta)\subseteq\operatorname{supp}(\Omega)$.
Then,
\begin{equation}
S^Q_\alpha(\rho\oplus\eta \,\|\, \sigma\oplus\Omega)
=
g^{-1}\!\left(
g\!\left(S^Q_\alpha(\rho \,\|\, \sigma)\right)
+
g\!\left(S^Q_\alpha(\eta \,\|\, \Omega)\right)
\right),
\label{eq:direct-sum-mean}
\end{equation}
where $g(x)=e^{(\alpha-1)x}$ and $g^{-1}(y)=\frac{1}{\alpha-1}\ln y$.
\end{proposition}

\begin{remark}
Note that $\rho\oplus\eta$ and $\sigma\oplus\Omega$ are not normalized states, 
since $\operatorname{Tr}(\rho\oplus\eta)=\operatorname{Tr}(\sigma\oplus\Omega)=2$. 
Nevertheless, the generating functional~\eqref{Geneq} remains well-defined 
because it depends only on the block-diagonal operator structure and the support 
conditions are satisfied on each component block.
\end{remark}

\begin{proof}
The operators $\rho\oplus\eta$ and $\sigma\oplus\Omega$ act on 
$\mathcal{H}_1\oplus\mathcal{H}_2$ and are block-diagonal. Hence,
\begin{equation}
\ln(\rho\oplus\eta)-\ln(\sigma\oplus\Omega)
=
(\ln\rho-\ln\sigma)\oplus(\ln\eta-\ln\Omega).
\end{equation}
Functional calculus preserves block structure, giving
\begin{equation}
e^{(\alpha-1)(\ln(\rho\oplus\eta)-\ln(\sigma\oplus\Omega))}
=
e^{(\alpha-1)(\ln\rho-\ln\sigma)}
\oplus
e^{(\alpha-1)(\ln\eta-\ln\Omega)}.
\end{equation}
Thus,
\begin{equation}
\begin{aligned}
&\Phi_{\rho\oplus\eta,\,\sigma\oplus\Omega}(\alpha-1) \\
&= \operatorname{Tr}\Bigl[ (\rho\oplus\eta) e^{(\alpha-1)\bigl(\ln(\rho\oplus\eta)-\ln(\sigma\oplus\Omega)\bigr)} \Bigr] \\
&= \operatorname{Tr}\bigl[ \rho\,e^{(\alpha-1)(\ln\rho-\ln\sigma)} \bigr] + \operatorname{Tr}\bigl[ \eta\,e^{(\alpha-1)(\ln\eta-\ln\Omega)} \bigr].
\end{aligned}
\end{equation}
By definition of $S^Q_\alpha$, this becomes
\begin{equation}
\Phi_{\rho\oplus\eta,\sigma\oplus\Omega}(\alpha-1)
=
e^{(\alpha-1)S^Q_\alpha(\rho\|\sigma)}
+
e^{(\alpha-1)S^Q_\alpha(\eta\|\Omega)}.
\end{equation}
Therefore,
\begin{equation}
\begin{aligned}
&S^Q_\alpha(\rho\oplus\eta \,\|\, \sigma\oplus\Omega) \\
&= \frac{1}{\alpha-1} \ln \left( e^{(\alpha-1)S^Q_\alpha(\rho\|\sigma)} + e^{(\alpha-1)S^Q_\alpha(\eta\|\Omega)} \right).
\end{aligned}
\end{equation}
which proves the statement.
\end{proof}

\begin{remark}[Normalized direct sum structure]
Let
\begin{equation}
    \tilde\rho=\frac{\rho\oplus\eta}{2}, \qquad
\tilde\sigma=\frac{\sigma\oplus\Omega}{2}.
\end{equation}
Since $\tilde\rho$ and $\tilde\sigma$ are normalized density 
operators, the logarithms are well defined. A direct computation 
gives
\begin{equation}
\begin{aligned}
\ln\tilde\rho - \ln\tilde\sigma
&= \bigl[ \ln(\rho\oplus\eta) - \ln 2 \cdot I \bigr] - \bigl[ \ln(\sigma\oplus\Omega) - \ln 2 \cdot I \bigr] \\
&= \ln(\rho\oplus\eta) - \ln(\sigma\oplus\Omega).
\end{aligned}
\end{equation}
since the scalar shift $-\ln 2\cdot I$ cancels in the difference.
Consequently, the generating functional of $(\tilde\rho,\tilde\sigma)$
satisfies
\begin{equation}
\begin{aligned}
\Phi_{\tilde\rho,\,\tilde\sigma}(\alpha-1)
&= \mathrm{Tr}\Bigl[ \tilde\rho\, e^{(\alpha-1)(\ln\tilde\rho-\ln\tilde\sigma)} \Bigr] \\
&= \frac{1}{2}\,\mathrm{Tr}\Bigl[ (\rho\oplus\eta)\, e^{(\alpha-1)\bigl(\ln(\rho\oplus\eta)-\ln(\sigma\oplus\Omega)\bigr)} \Bigr].
\end{aligned}
\end{equation}
where the factor $\tfrac{1}{2}$ arises from 
$\tilde\rho = \tfrac{1}{2}(\rho\oplus\eta)$, 
and not from any failure of normalization of $\tilde\rho$ itself.
Using the block-diagonal structure of $\rho\oplus\eta$ and 
$\sigma\oplus\Omega$, this reduces to
\begin{equation}
\Phi_{\tilde\rho,\tilde\sigma}(\alpha-1)
=
\frac{1}{2}
\left(
e^{(\alpha-1)S^Q_\alpha(\rho\|\sigma)}
+
e^{(\alpha-1)S^Q_\alpha(\eta\|\Omega)}
\right),
\end{equation}
and therefore
\begin{equation}
S^Q_\alpha(\tilde\rho\|\tilde\sigma)
=
\frac{1}{\alpha-1}
\ln\!\left[
\frac{1}{2}
\left(
e^{(\alpha-1)S^Q_\alpha(\rho\|\sigma)}
+
e^{(\alpha-1)S^Q_\alpha(\eta\|\Omega)}
\right)
\right].
\end{equation}
Equivalently, in terms of the generating function 
$g(x)=e^{(\alpha-1)x}$,
\begin{equation}
S^Q_\alpha(\tilde\rho\|\tilde\sigma)
=
g^{-1}\!\left(
\frac{1}{2}g(x)+\frac{1}{2}g(y)
\right),
\end{equation}
where $x=S^Q_\alpha(\rho\|\sigma)$ and 
$y=S^Q_\alpha(\eta\|\Omega)$.
This shows that the normalized direct sum induces a 
Kolmogorov-Nagumo mean with equal weights $\tfrac{1}{2}$ 
under the generating function $g(x)=e^{(\alpha-1)x}$, 
corresponding to a log-sum-exp aggregation in the transformed 
$g$-space.
\end{remark}

\begin{proposition}[Reduction to the Umegaki quantum relative entropy]\label{prop:classical-reduction}
Let $(\rho,\sigma) \in \mathcal{D}(\mathcal{H}) \times \mathcal{D}(\mathcal{H})$
be density operators satisfying 
$\operatorname{supp}(\rho)\subseteq\operatorname{supp}(\sigma)$.
Then the Cu-Q relative Rényi functional $S_\alpha^Q(\rho\|\sigma)$
reduces to the Umegaki quantum relative entropy in the limit $\alpha\to1^+$:
\begin{equation}
\lim_{\alpha\to1} S^Q_\alpha(\rho\|\sigma)
=
S(\rho\|\sigma) = \operatorname{Tr}\!\left[\rho(\ln\rho-\ln\sigma)\right].
\end{equation}
\end{proposition}

\begin{proof}
Since $\tau=\alpha-1$, the limit
$\alpha\to1^+$ is equivalent to $\tau\to0^+$. The generating functional~\eqref{Geneq} becomes
\begin{equation}
    \Phi_{\rho,\sigma}(\tau)=\operatorname{Tr}\!\left[
    \rho\,e^{\tau(\ln\rho-\ln\sigma)}\right].
\end{equation}
Expanding the operator exponential around $\tau=0$ gives
\begin{equation}
    e^{\tau(\ln\rho-\ln\sigma)}=\mathbb I +\tau(\ln\rho-\ln\sigma)+\frac{\tau^2}{2}(\ln\rho-\ln\sigma)^2+\mathcal O(\tau^3).
\end{equation}
Substituting this expansion into the trace yields
\begin{equation}
    \Phi_{\rho,\sigma}(\tau)=\operatorname{Tr}(\rho)+
    \tau\operatorname{Tr}[\rho(\ln\rho-\ln\sigma)]+
    \mathcal O(\tau^2).
\end{equation}
Since $\operatorname{Tr}(\rho)=1$,
\begin{equation}
\Phi_{\rho,\sigma}(\tau)
=
1+\tau S(\rho\|\sigma)
+\mathcal O(\tau^2),
\end{equation}
where \(S(\rho\|\sigma)=\operatorname{Tr}\!\left[
\rho(\ln\rho-\ln\sigma)\right]\)
is the Umegaki quantum relative entropy.
With the expansion \(\ln(1+x)=x+\mathcal O(x^2)\) as \(x\to0\), we obtain
\begin{equation}
    \ln\Phi_{\rho,\sigma}(\tau)=\tau S(\rho\|\sigma)+
    \mathcal O(\tau^2).
\end{equation}
Therefore,
\begin{equation}
    S_\alpha^Q(\rho\|\sigma)=\frac{1}{\tau} \ln\Phi_{\rho,\sigma}(\tau)=S(\rho\|\sigma)+
    \mathcal O(\tau).
\end{equation}
Hence \(\lim_{\tau\to0} S_\alpha^Q(\rho\|\sigma)= S(\rho\|\sigma). \)
\end{proof}

The limit $\alpha\to1$ shows that the Cu-Q relative Rényi functional is consistent with the Umegaki quantum relative entropy. A complementary consistency check is provided by the commuting case, where quantum
expressions should reduce to their classical counterparts. We now show that the proposed functional recovers the classical Rényi divergence
whenever $\rho$ and $\sigma$ commute.

\begin{proposition}[Classical Consistency]
Let $\alpha >1$ and let
$(\rho,\sigma) \in \mathcal{D}(\mathcal{H}) \times \mathcal{D}(\mathcal{H})$
be commuting density operators, i.e., $[\rho,\sigma]=0$, satisfying
$\operatorname{supp}(\rho)\subseteq\operatorname{supp}(\sigma)$.
Then,
\begin{equation}\label{eq:classical-reduction}
S^Q_\alpha(\rho\|\sigma) =
\frac{1}{\alpha-1}\ln\operatorname{Tr}\!\left[
\rho^\alpha \sigma^{1-\alpha}
\right].
\end{equation}
\end{proposition}

\begin{proof}
Since $[\rho,\sigma]=0$, the operators admit a common orthonormal
eigenbasis $\{|\varphi_i\rangle\}_{i=1}^d$ such that
\begin{equation}
    \rho=\sum_i
    p(\varphi_i)\,|\varphi_i\rangle\langle\varphi_i|,
    \quad\sigma=\sum_i q(\varphi_i)\,|\varphi_i\rangle\langle\varphi_i|.
\end{equation}
The support condition
$\operatorname{supp}(\rho)\subseteq\operatorname{supp}(\sigma)$ implies that $q(\varphi_i)>0$
whenever $p(\varphi_i)>0$. Applying functional calculus,
\begin{equation}
    \ln\rho
    =
    \sum_i
    \ln p(\varphi_i)\,
    |\varphi_i\rangle\langle\varphi_i|,
    \quad
    \ln\sigma
    =
    \sum_i
    \ln q(\varphi_i)\,
    |\varphi_i\rangle\langle\varphi_i|.
\end{equation}
Hence
\begin{equation}
    e^{(\alpha-1)(\ln\rho-\ln\sigma)}=\sum_i
    p(\varphi_i)^{\alpha-1} q(\varphi_i)^{1-\alpha}
    |\varphi_i\rangle\langle\varphi_i|.
\end{equation}
Therefore,
\begin{equation}
    \operatorname{Tr}\!\left[\rho\, e^{(\alpha-1)(\ln\rho-\ln\sigma)}\right]=\sum_i p(\varphi_i)^\alpha q(\varphi_i)^{1-\alpha}.
\end{equation}
Similarly,
\begin{equation}
    \rho^\alpha\sigma^{1-\alpha}=\sum_i p(\varphi_i)^\alpha q(\varphi_i)^{1-\alpha}
    |\varphi_i\rangle\langle\varphi_i|,
\end{equation}
so that
\begin{equation}
    \operatorname{Tr}\!\left[\rho^\alpha\sigma^{1-\alpha}\right]=\sum_i p(\varphi_i)^\alpha
    q(\varphi_i)^{1-\alpha}.
\end{equation}
Combining the two identities yields
\begin{equation}
    \operatorname{Tr}\!\left[\rho\,e^{(\alpha-1)(\ln\rho-\ln\sigma)}\right]= \operatorname{Tr}\!\left[\rho^\alpha\sigma^{1-\alpha}\right].
\end{equation}
Substituting this relation into
Definition~\ref{QCRRF} gives
\[
S_\alpha^Q(\rho\|\sigma)
=
\frac{1}{\alpha-1}
\ln
\operatorname{Tr}\!\left[
\rho^\alpha\sigma^{1-\alpha}
\right],
\]
which proves \eqref{eq:classical-reduction}.
\end{proof}

The preceding results establish consistency with two fundamental limiting cases: the Cu-Q relative Rényi functional reduces to the Umegaki quantum relative entropy in the limit $\alpha\to1$ and recovers the classical Rényi divergence whenever the density operators commute. Having verified these consistency requirements, we now turn to intrinsic properties of the functional itself, focusing on positivity and its dependence on the Rényi parameter $\alpha$.

\begin{theorem}[Positivity for $\alpha > 1$]
Let $\alpha > 1$ and let
$(\rho,\sigma) \in \mathcal{D}(\mathcal{H}) \times \mathcal{D}(\mathcal{H})$
be density operators satisfying
$\operatorname{supp}(\rho)\subseteq\operatorname{supp}(\sigma)$.
Then,
\begin{equation}
    S_\alpha^{Q}(\rho\|\sigma) \ge 0.
\end{equation}
\end{theorem}

\begin{proof}
By the definition of $S^Q_\alpha(\rho\|\sigma)$,
\begin{equation}
    S_\alpha^{Q}(\rho\|\sigma)
    =\frac{1}{\alpha-1}\ln \operatorname{Tr}\!\left[
    \rho\, e^{(\alpha-1)(\ln\rho-\ln\sigma)}\right].
\end{equation}
By the Peierls-Bogoliubov inequality 
(see, e.g., \cite{Wehrl1978GeneralEntropy,bhatia1997matrix}), 
for any density operator $\rho$ and Hermitian operator $X$,
\begin{equation}
\operatorname{Tr}[\rho\, e^{X}] \ge e^{\operatorname{Tr}[\rho X]},
\end{equation}
with equality if and only if $X$ is constant on $\operatorname{supp}(\rho)$.
Applying this with $X = (\alpha-1)(\ln\rho-\ln\sigma)$ yields
\begin{equation}
\operatorname{Tr}\!\left[\rho\, e^{(\alpha-1)(\ln\rho-\ln\sigma)}\right]
\ge 
e^{(\alpha-1)\operatorname{Tr}[\rho(\ln\rho-\ln\sigma)]}.
\end{equation}
By Klein's inequality \cite{bhatia1997matrix,Carlen2010}, $\operatorname{Tr}[\rho(\ln\rho-\ln\sigma)] \ge 0$.
Since $\alpha - 1 > 0$, we have
$(\alpha-1)\operatorname{Tr}[\rho(\ln\rho-\ln\sigma)] \ge 0$,
and therefore
\begin{equation}
e^{(\alpha-1)\operatorname{Tr}[\rho(\ln\rho-\ln\sigma)]} \ge e^0 = 1.
\end{equation}
Combining the above,
\begin{equation}
\operatorname{Tr}\!\left[
\rho\, e^{(\alpha-1)(\ln\rho-\ln\sigma)}
\right] \ge 1.
\end{equation}
Since $\operatorname{Tr}[\rho\, e^{(\alpha-1)(\ln\rho-\ln\sigma)}] \ge 1$,
taking the logarithm gives
\begin{equation}
\ln\operatorname{Tr}\!\left[\rho\, e^{(\alpha-1)(\ln\rho-\ln\sigma)}\right] \ge 0.
\end{equation}
Since $\alpha - 1 > 0$, dividing by $\alpha - 1$ preserves the inequality,
and therefore
\begin{equation}
S^Q_\alpha(\rho\|\sigma) 
= \frac{1}{\alpha-1}
\ln\operatorname{Tr}\!\left[\rho\, e^{(\alpha-1)(\ln\rho-\ln\sigma)}\right]
\ge 0.
\end{equation}
\end{proof}

Having established positivity on the studied regime $\alpha>1$, we now turn to the final property of this section: monotonicity of $S_\alpha^Q(\rho\|\sigma)$ with respect to the Rényi parameter $\alpha$. This property characterizes how the distinguishability between $\rho$ and $\sigma$, as quantified by the proposed functional, varies with the parameter $\alpha$.
 
\begin{theorem}[Monotonicity in the Rényi parameter]\label{Th10}
Let $ \alpha \ge \gamma > 1,$ and let $(\rho,\sigma)\in \mathcal D(\mathcal H)\times\mathcal D(\mathcal H)$ be density operators satisfy $\operatorname{supp}(\rho)\subseteq\operatorname{supp}(\sigma).$
Then
\begin{equation}
S_\gamma^Q(\rho\|\sigma)
\le
S_\alpha^Q(\rho\|\sigma).
\end{equation}
Hence, the Cu-Q relative Rényi functional is monotonically
non-decreasing in $\alpha$ on $(1,\infty)$.
\end{theorem}

\begin{proof}
Starting from \eqref{QCRRF} and \eqref{Geneq}, consider the quantum relative surprisal operator
$\Delta\Xi=\ln\rho-\ln\sigma,$ which is a Hermitian operator. By the spectral theorem, there exists a unitary operator such that
\begin{equation}
\Delta\Xi = U \left(\sum_j \lambda_j |v_j\rangle\langle v_j| \right) U^\dagger .
\end{equation}
Consequently,
\begin{equation}
e^{(\alpha-1)\Delta\Xi} = U \left(\sum_j e^{(\alpha-1)\lambda_j}|v_j\rangle\langle v_j|\right)U^\dagger .
\end{equation}

Using the spectral decomposition of $\rho$ in~\eqref{rho}, substituting into the definition of $S_\alpha^Q(\rho\|\sigma)$, and applying the cyclicity of the trace, we obtain
\begin{equation}
    \operatorname{Tr}\!\left[\rho\,e^{(\alpha-1)\Delta\Xi}\right]=\sum_{i,j}p(\varphi_i)\,e^{(\alpha1)\lambda_j}\left|\langle\varphi_i|U|v_j\rangle \right|^2 .
\end{equation}
Introducing the notation $|e_j\rangle:=U|v_j\rangle$ and collecting the terms associated with the same eigenvalue $\lambda_j$, we may rewrite the above expression as
\begin{equation}
    \operatorname{Tr}\!\left[\rho\,e^{(\alpha 1)\Delta\Xi}\right]=\sum_jp'(e_j)\,e^{(\alpha-1)\lambda_j},
\end{equation}
where
\begin{equation}
p'(e_j) := \sum_i p(\varphi_i) \left| \langle e_j|\varphi_i\rangle \right|^2.
\end{equation}
Since $p(\varphi_i) \ge 0$ for all $i$ and 
$\sum_i p(\varphi_i) = 1$, the weights $p'(e_j)$ are non-negative.
To verify normalization, we use the completeness relation
$\sum_j |e_j\rangle\langle e_j| = \mathbb{I}$ to obtain
\begin{equation}
    \sum_j p'(e_j) = \sum_{i,j} p(\varphi_i)
    \left|\langle e_j|\varphi_i\rangle\right|^2 = 1.
\end{equation}
Hence $\{p'(e_j)\}$ is a probability distribution.Therefore,
\begin{equation}
    S_\alpha^Q(\rho\|\sigma)=\frac{1}{\alpha-1}\ln\sum_j p'(e_j)\,e^{(\alpha-1)\lambda_j}.
\end{equation}

To compare
$S_\alpha^Q(\rho\|\sigma)$ and $S_\gamma^Q(\rho\|\sigma)$,
define
\begin{equation}
m := \frac{\alpha-1}{\gamma-1}.
\end{equation}
Because $\alpha\ge\gamma>1$, we have $\alpha-1\ge\gamma-1>0.$ Dividing by the positive quantity $\gamma-1$ gives
\begin{equation}
m=\frac{\alpha-1}{\gamma-1}\ge1.
\end{equation}
Writing $\alpha-1=m(\gamma-1),$ we obtain
\begin{equation}
    e^{(\alpha-1)\lambda_j}=e^{m(\gamma-1)\lambda_j}=
    \left(e^{(\gamma-1)\lambda_j}\right)^m.
\end{equation}
Define
\begin{equation}
    R_j := e^{(\gamma-1)\lambda_j}>0.
\end{equation}
Then
\begin{equation}
\sum_j p'(e_j)\, e^{(\alpha-1)\lambda_j} = \sum_j p'(e_j)\, R_j^{\,m}=\mathbb E[R^m],
\end{equation}
where the expectation is taken with respect to the probability distribution
$\{p'(e_j)\}$.

Since $\phi(x)=x^m$ satisfies
\begin{equation}
    \phi''(x)=m(m-1)x^{m-2}\ge0,
\end{equation}
it is convex on $(0,\infty)$. Taking logarithms and using $m=\frac{\alpha-1}{\gamma-1},$ we obtain
\begin{equation}
\ln\sum_jp'(e_j)e^{(\alpha-1)\lambda_j}\ge \frac{\alpha-1}{\gamma-1}\ln\sum_jp'(e_j)e^{(\gamma-1)\lambda_j}.
\end{equation}

Since $\alpha>1$, division by $\alpha-1$ preserves the inequality and yields
\begin{equation}
\frac{1}{\alpha-1}\ln\sum_jp'(e_j)e^{(\alpha-1)\lambda_j}\ge\frac{1}{\gamma-1}\ln\sum_jp'(e_j)
e^{(\gamma-1)\lambda_j}.
\end{equation}
Therefore,
\begin{equation}
S_\alpha^Q(\rho\|\sigma)
\ge S_\gamma^Q(\rho\|\sigma),
\end{equation}
which proves the claim.
\end{proof}

The preceding results were established for the non-regularized Cu-Q relative Rényi functional on its studied regime $\alpha>1$. We now consider the regularized formulation \eqref{regularized_Cu-Q}, which extends the definition to all $\alpha\in\mathbb R\setminus\{1\}$.

The regularization leaves the fundamental structure of the functional unchanged. Consequently, the structural properties established for the non-regularized formulation continue to hold in the regularized setting, while the enlarged parameter domain makes it possible to investigate additional features beyond the studied regime.

In particular, we establish a global monotonicity property with respect to the Rényi parameter, examine the behaviour at $\alpha=0$, and develop an interpretation of the functional in terms of relative
quantumness measure.

\section{Extended parameter analysis and relative quantumness} \label{Relative Quantumness}

We first investigate the dependence of the regularized Cu-Q relative Rényi functional on the Rényi parameter. The following result shows that the functional remains monotonically non-decreasing throughout its extended parameter domain and therefore admits no turning point

\begin{corollary}[Global monotonicity of the regularized Cu-Q relative Rényi functional]
Let $(\rho,\sigma)\in\mathcal D(\mathcal H) \times \mathcal D(\mathcal H)$ and let $ \rho_\varepsilon,\sigma_\varepsilon $ be the regularized states defined in Eq.~\eqref{state_regularized}. Then
\begin{equation}
    \frac{d}{d\alpha} S_\alpha^Q(\rho_\varepsilon\|\sigma_\varepsilon)\ge 0
\end{equation}
for every $\alpha\in\mathbb R\setminus\{1\}.$ Consequently, the regularized Cu-Q relative Rényi functional is
monotonically non-decreasing with respect to the Rényi parameter throughout its entire parameter domain. In particular, the function admits no turning point.
\end{corollary}

\begin{proof}
The argument follows the same steps as in the proof of the monotonicity theorem for the non-regularized functional. Let 
\begin{equation}
    \Delta\Xi_\varepsilon =\ln\rho_\varepsilon -\ln\sigma_\varepsilon
\end{equation}
be the regularized quantum relative surprisal operator, and let
\begin{equation}
\Delta\Xi_\varepsilon=U\left(\sum_j\tilde{\lambda}_j|\tilde v_j\rangle\langle\tilde v_j|\right)U^\dagger
\end{equation}
be its spectral decomposition. 
Define $|e_j\rangle := U|\tilde{v}_j\rangle.$ We then substitute the spectral decompositions of $\rho_\varepsilon$
and $e^{(\alpha-1)\Delta\Xi_\varepsilon}$ into the trace expression and use the cyclicity of the trace to obtain
\begin{equation}
\operatorname{Tr}
\!\left[
\rho_\varepsilon
e^{(\alpha-1)\Delta\Xi_\varepsilon}
\right]
=
\sum_j
p'_\varepsilon(e_j)
e^{(\alpha-1)\tilde{\lambda}_j},
\end{equation}
where
\begin{equation}
p'_\varepsilon(e_j)
:=
\sum_i
p_\varepsilon(\varphi_i)
\left|
\langle e_j|\varphi_i\rangle
\right|^2 .
\end{equation}
Substituting this expression into the definition of
$S_\alpha^Q(\rho_\varepsilon\|\sigma_\varepsilon)$
yields
\begin{equation}
S_\alpha^Q(\rho_\varepsilon\|\sigma_\varepsilon)
=
\frac{\ln\kappa(\alpha)}
{\alpha-1},
\end{equation}
where
\begin{equation}
\kappa(\alpha) :=\sum_j p'_\varepsilon(e_j)e^{(\alpha-1)\tilde{\lambda}_j}.
\end{equation}
Because $\rho_\varepsilon\ge\varepsilon \mathbb{\mathbb{I}}/d,$ we have $ p'_\varepsilon(e_j) = \langle e_j|\rho_\varepsilon|e_j\rangle\ge\varepsilon/d
>0.$ Moreover, $\sum_j p'_\varepsilon(e_j)  = 1,$ and hence $\{p'_\varepsilon(e_j)\}$
is a strictly positive probability distribution.

Define the exponential tilting of $\{p_\varepsilon '(e_j)\}$,
\begin{equation}
q_j(\alpha) := \frac{p_\varepsilon '(e_j)e^{(\alpha -1)\tilde{\lambda}_j}}{\kappa(\alpha)}.
\end{equation}
and the mean of $\{\tilde{\lambda}_j\}$ under $q_j(\alpha)$,
\begin{equation}
\mu(\alpha) := \sum_j q_j(\alpha)\tilde{\lambda}_j.
\end{equation}
Since $p_\varepsilon '(e_j)>0$ and $ e^{(\alpha -1)\tilde{\lambda}_j}>0,$ we have $q_j(\alpha)>0.$
Moreover,
\begin{equation}
\sum_j q_j(\alpha)=\frac{\sum_jp_\varepsilon '(e_j)
e^{(\alpha-1)\tilde{\lambda}_j}}{\kappa(\alpha)}=1.
\end{equation}
Hence $\{q_j(\alpha)\}$ is also a strictly positive probability distribution.

Differentiating $S_\alpha^Q(\rho_\varepsilon\|\sigma_\varepsilon)$ via the quotient rule and using
\begin{equation}
\frac{\kappa'(\alpha)}{\kappa(\alpha)} 
= \frac{\sum_j p_\varepsilon '(e_j)\,\tilde{\lambda}_j\, e^{(\alpha-1)\tilde{\lambda}_j}}{\kappa(\alpha)} 
= \sum_j q_j(\alpha)\tilde{\lambda}_j = \mu(\alpha),
\end{equation}
we obtain
\begin{equation}
\frac{d}{d\alpha}S_\alpha^Q(\rho_\varepsilon\|\sigma_\varepsilon)
= \frac{(\alpha-1)\mu(\alpha) - \ln\kappa(\alpha)}{(\alpha-1)^2}.
\end{equation}

Since $(\alpha-1)^2 > 0$, it suffices to show that the numerator is non-negative. Taking the logarithm of the definition of $q_j(\alpha)$,
\begin{equation}
\ln q_j(\alpha) = \ln p_\varepsilon '(e_j) + (\alpha-1)\tilde{\lambda}_j - \ln\kappa(\alpha),
\end{equation}
and rearranging,
\begin{equation}
(\alpha-1)\tilde{\lambda}_j - \ln\kappa(\alpha) = \ln q_j(\alpha) - \ln p_\varepsilon '(e_j).
\end{equation}
Multiplying both sides by $q_j(\alpha)$ and summing over $j$,
\begin{equation}
\begin{aligned}
    &(\alpha-1)\underbrace{\sum_j q_j(\alpha)\tilde{\lambda}_j}_{=\,\mu(\alpha)} - \ln\kappa(\alpha)\underbrace{\sum_j q_j(\alpha)}_{=\,1} \\
    &= \sum_j q_j(\alpha)\bigl[\ln q_j(\alpha) - \ln p_\varepsilon '(e_j)\bigr] \\
    &= \sum_j q_j(\alpha)\ln\frac{q_j(\alpha)}{p_\varepsilon '(e_j)}.
\end{aligned}
\end{equation}
The right-hand side is precisely the relative entropy $S(q(\alpha)\|p_\varepsilon')$, which is non-negative by the Gibbs inequality, 
with equality if and only if $q_j(\alpha) = p_\varepsilon '(e_j)$ for all $j$.

Therefore,
\begin{equation}
\frac{d}{d\alpha}S^Q_\alpha(\rho_\varepsilon\|\sigma_\varepsilon) 
= \frac{S(q(\alpha)\|p_\varepsilon')}{(\alpha-1)^2} \ge 0,
\end{equation}
confirming that $S^Q_\alpha(\rho_\varepsilon\|\sigma_\varepsilon)$ is monotonically non-decreasing 
in $\alpha$ on $\alpha\in\mathbb R\setminus\{1\}$.
Consequently, the regularized Cu-Q relative Rényi functional admits no turning point.
\end{proof}

We now turn to the special case $\alpha=0$. The following theorem establishes that the regularized order-zero Cu-Q relative Rényi functional is always non-positive.

\begin{theorem}[$\alpha=0$ regularized Cu-Q relative Rényi functional]\label{Th12}
Let $(\rho,\sigma)\in\mathcal D(\mathcal H) \times \mathcal D(\mathcal H)$ and let $\rho_\varepsilon,\sigma_\varepsilon$
be the regularized states defined in Eq.\eqref{state_regularized}.
Then
\begin{equation}
    S_0^Q(\rho_\varepsilon\|\sigma_\varepsilon)
    \le 0.
\end{equation}
\end{theorem}

\begin{proof}
Setting $\alpha=0$ in Eq.\eqref{regularized_Cu-Q}
gives
\begin{equation}
    S_0^Q(\rho_\varepsilon\|\sigma_\varepsilon)=-\ln \operatorname{Tr} \!\left[\rho_\varepsilon e^{\ln\sigma_\varepsilon-\ln\rho_\varepsilon}\right].
\end{equation}
Since
$\rho_\varepsilon=e^{\ln\rho_\varepsilon}$, we may write
\begin{equation}
    \operatorname{Tr}\!\left[\rho_\varepsilon e^{\ln\sigma_\varepsilon-\ln\rho_\varepsilon}
    \right] =\operatorname{Tr}\!\left[e^{\ln\rho_\varepsilon}
    e^{\ln\sigma_\varepsilon-\ln\rho_\varepsilon}\right].
\end{equation}
Applying the Golden-Thompson inequality,
\begin{equation}
    \operatorname{Tr}(e^{A+B})\le \operatorname{Tr}(e^A e^B),
\end{equation}
with $A=\ln\rho_\varepsilon,\quad B=\ln\sigma_\varepsilon-\ln\rho_\varepsilon$,
yields
\begin{equation}
    \operatorname{Tr} \!\left[e^{\ln\sigma_\varepsilon}
    \right] \le\operatorname{Tr}\!\left[
        e^{\ln\rho_\varepsilon}
        e^{\ln\sigma_\varepsilon-\ln\rho_\varepsilon}
    \right].
\end{equation}

Using $e^{\ln\sigma_\varepsilon} = \sigma_\varepsilon$
together with $\operatorname{Tr}(\sigma_\varepsilon)=1,$
we obtain
\begin{equation}
    \operatorname{Tr}\!\left[\rho_\varepsilon
        e^{\ln\sigma_\varepsilon-\ln\rho_\varepsilon}
    \right]\ge 1.
\end{equation}
Because the function $-\ln x$ is decreasing on $(0,\infty)$,
\begin{equation}
    S_0^Q(\rho_\varepsilon\|\sigma_\varepsilon)=-\ln
    \operatorname{Tr}\!\left[\rho_\varepsilon
        e^{\ln\sigma_\varepsilon-\ln\rho_\varepsilon}
    \right] \le-\ln 1=0.
\end{equation}
This completes the proof.
\end{proof}

The previous results characterize the behavior of the regularized Cu-Q relative Rényi functional as a function of the Rényi parameter. We now investigate how the functional reflects the non-commutativity between the underlying quantum states.

To this end, we introduce a notion of relative quantumness and establish a direct connection between this quantity and the regularized Cu-Q relative Rényi functional.

\begin{theorem}[Regularized relative quantumness]\label{quantumness}

Let $(\rho,\sigma)\in\mathcal{D}(\mathcal{H})\times\mathcal{D}(\mathcal{H})$,
and let $\rho_\varepsilon,\sigma_\varepsilon$ be the regularized states. Define
\begin{equation}
  Q(\rho_\varepsilon\|\sigma_\varepsilon) := -S_0^Q(\rho_\varepsilon\|\sigma_\varepsilon).
\end{equation}
Then
\begin{equation}
  Q(\rho_\varepsilon\|\sigma_\varepsilon)
  \begin{cases}
    = 0, & \text{iff } [\rho,\sigma]=0,\\[6pt]
    > 0, & \text{iff } [\rho,\sigma]\neq 0.
  \end{cases}
\end{equation}
\end{theorem}

\begin{proof}
Since $\rho_\varepsilon=(1-\varepsilon)\rho+ \varepsilon\frac{\mathbb{I}}{d},\quad\sigma_\varepsilon=(1-\varepsilon)\sigma+\varepsilon\frac{\mathbb{I}}{d},$
and the identity operator commutes with every operator, we have
\begin{equation}
[\rho_\varepsilon,\sigma_\varepsilon]=(1-\varepsilon)^2 [\rho,\sigma].
\label{eq:comm-equivalence}
\end{equation}
Hence $[\rho_\varepsilon,\sigma_\varepsilon]=0$ if and only if $[\rho,\sigma]=0$.

From Theorem \ref{Th12},
\begin{equation}
S_0^Q(\rho_\varepsilon\|\sigma_\varepsilon)=-\ln \operatorname{Tr}\!\left[\rho_\varepsilon
e^{\ln\sigma_\varepsilon-\ln\rho_\varepsilon}\right].
\end{equation}

\vspace{0.3 cm}
\noindent
\textit{(i) Commuting case.}
Assume that $[\rho_\varepsilon,\sigma_\varepsilon]=0.$ Then
$[\ln\rho_\varepsilon,\ln\sigma_\varepsilon]=0,$
because commutativity is preserved under analytic functional calculus. Thus
\begin{equation}
e^{\ln\sigma_\varepsilon-\ln\rho_\varepsilon} = \sigma_\varepsilon\rho_\varepsilon^{-1}.
\end{equation}

Using the cyclicity property of the trace, we obtain
\begin{equation}
\operatorname{Tr}\!\left[\rho_\varepsilon\sigma_\varepsilon
\rho_\varepsilon^{-1}\right]=\operatorname{Tr}(\sigma_\varepsilon)=1.
\end{equation}
It follows that $S_0^Q(\rho_\varepsilon\|\sigma_\varepsilon)
=-\ln1=0,$ and therefore $Q(\rho_\varepsilon\|\sigma_\varepsilon)=0.$

\vspace{0.3 cm}
\noindent
\textit{(ii) Non-commuting case.}
Assume that $[\rho_\varepsilon,\sigma_\varepsilon]\neq0.$ Then $[\ln\rho_\varepsilon,\ln\sigma_\varepsilon]\neq0,$
since otherwise $\rho_\varepsilon$ and $\sigma_\varepsilon$
would be simultaneously diagonalizable. 

By the Golden-Thompson inequality, with equality if and only if $[A,B]=0$ \cite{Petz1994,So1992}, and since $[A,B]\neq0,$
we obtain
\begin{equation}
\operatorname{Tr}(e^{A+B})
<\operatorname{Tr}(e^Ae^B).
\end{equation}
Substituting $e^A=\rho_\varepsilon,$ and $e^{A+B}=\sigma_\varepsilon, $ yields
\begin{equation}
1=\operatorname{Tr}(\sigma_\varepsilon)
<\operatorname{Tr}\!\left[\rho_\varepsilon
e^{\ln\sigma_\varepsilon-\ln\rho_\varepsilon}
\right].
\end{equation}
Taking the negative logarithm gives $S_0^Q(\rho_\varepsilon\|\sigma_\varepsilon)<0.$ Therefore $Q(\rho_\varepsilon\|\sigma_\varepsilon)=-S_0^Q(\rho_\varepsilon\|\sigma_\varepsilon)
>0.$
\end{proof}

The above theorem establishes that the negative regularized relative Rényi functional at $\alpha=0$, denoted by $Q(\rho_\varepsilon\|\sigma_\varepsilon) =-S_0^Q(\rho_\varepsilon\|\sigma_\varepsilon)$, provides a complete characterization of commutativity between quantum states. Specifically, $-S_0^Q(\rho_\varepsilon\|\sigma_\varepsilon) = 0$ if and only if the states $\rho$ and $\sigma$ commute.  Conversely, $-S_0^Q(\rho_\varepsilon\|\sigma_\varepsilon) > 0$ if and only if the states do not commute.

In this sense, $Q(\rho_\varepsilon\|\sigma_\varepsilon)$ may be interpreted as a measure of relative quantumness between $\rho_\varepsilon$ and $\sigma_\varepsilon$, quantifying the degree to which the two states fail to commute. Unlike basis-dependent notions of quantumness, such as coherence measures, or correlation-based quantities such as quantum discord, the present functional is defined directly from the ordered pair $(\rho_\varepsilon,\sigma_\varepsilon)$
and depends solely on their mutual non-commutativity.

\section{The QDPI Conjecture and Numerical Investigation on Commutativity-Preserving Channels}\label{CoP-QDPI}

\subsection{CoP Channels and the QDPI Conjecture}
The relative quantumness $Q(\rho_\varepsilon\|\sigma_\varepsilon)$
naturally raises the question of whether it obeys a data-processing principle under physically relevant quantum evolutions.
In this section, we focus on a particular class of CPTP maps, namely commutativity-preserving (CoP) channels, which map
commuting input states to commuting output states. Such channels were introduced by Yu et al. \cite{Yu2011} and further studied by Guo et al. \cite{Guo_2013}. Related structural properties of noisy quantum channels and their effects on quantum correlations were investigated in Ref.~\cite{Streltsov2011}. We begin by recalling the definition of a CoP channel and then formulate a conjectured QDPI governing the evolution of relative quantumness under such channels.

\begin{definition}[Commutativity-preserving channel]\label{def:comm-preserving}
Let $\mathcal{H}$ be a finite-dimensional Hilbert space. A CPTP map $\mathcal{N}:\mathcal{B}(\mathcal{H})\to\mathcal{B}(\mathcal{H})$ is called a CoP channel if, for every
pair $(\rho,\sigma)\in\mathcal{D}(\mathcal{H})\times
\mathcal{D}(\mathcal{H})$ satisfying $[\rho,\sigma]=0$, one has $[\mathcal {N}(\rho),\mathcal {N}(\sigma)]=0$.
\end{definition}

Motivated by the interpretation of
$Q(\rho_\varepsilon\|\sigma_\varepsilon)$
as a measure of relative quantumness, we now investigate its behavior under CoP channels. The central question is whether this quantity is monotone under such channels, which leads to the following conjectured quantum data-processing inequality.

Nevertheless, although the regularization procedure guarantees that $\rho_\varepsilon$ and $\sigma_\varepsilon$ are strictly positive; their images under a quantum channel need not remain full-rank. Consequently, the logarithmic operators appearing in $Q\left(\mathcal N(\rho_\varepsilon)\middle\|\mathcal N(\sigma_\varepsilon)\right)$ may not be globally defined on the entire Hilbert space. To ensure that the functional remains well-defined, we impose the support condition $\operatorname{supp}\!\bigl(\mathcal N(\rho_\varepsilon)\bigr)\subseteq\operatorname{supp}\!\bigl(\mathcal N(\sigma_\varepsilon)
\bigr).$ Under this assumption, the regularized relative quantumness remains well-defined even when the channel outputs are rank-deficient. This leads naturally to the following conjectured QDPI-type monotonicity relation.

\begin{conjecture}[QDPI at $\alpha=0$ for non-commuting inputs under CoP channels]
\label{conj:QDPI-alpha0}
Let $\rho,\sigma)\in \mathcal D(\mathcal H)\times\mathcal D(\mathcal H)$
satisfy $[\rho,\sigma]\neq0,$ where $\rho_\varepsilon,\sigma_\varepsilon$
are the regularized states defined in Eq.\eqref{state_regularized}.
Suppose that $\mathcal N $ is a CoP channel satisfying $\operatorname{supp}\!\bigl(\mathcal N(\rho_\varepsilon)\bigr) \subseteq\operatorname{supp}\!\bigl(\mathcal N(\sigma_\varepsilon)\bigr). $
Then,
\begin{equation}
    Q(\rho_\varepsilon\|\sigma_\varepsilon)
    \geq
    Q\left(
        \mathcal N(\rho_\varepsilon)\|\,
        \mathcal N(\sigma_\varepsilon)
    \right) \geq 0.
\end{equation}
\end{conjecture}

The present study examines several representative quantum channels, including depolarizing, dephasing, bit-flip, phase-flip, semi-classical, and transpose-type isotropic channels. These channels provide a diverse testing ground for investigating the conjectured QDPI behavior of the regularized relative quantumness under CoP evolutions.

For qubit systems, CoP channels admit a complete characterization: every CoP channel is either unital or semi-classical. Within this classification, the depolarizing, dephasing, and Pauli-type flip channels belong to the unital class, while semi-classical channels constitute the second class. Consequently, all of these channels are CoP in dimension two.

In higher dimensions ($\dim\mathcal H\ge 3$), the structure of CoP channels becomes considerably more restrictive. As shown in Refs.~\cite{Yu2011,Guo_2013}, every CoP channel belongs to one of only two classes: (i) completely decohering channels and (ii) nontrivial isotropic channels. Among the channels considered here, the semi-classical channel represents the completely decohering class, whereas the transpose-type isotropic channel belongs to the nontrivial isotropic class. The depolarizing channel may be viewed as a special instance of an isotropic channel. By contrast, a generic dephasing channel does not belong to either class except in the fully decohering limit and therefore is generally not CoP for $\dim\mathcal H\ge 3$. Similarly, higher-dimensional analogues of bit-flip and phase-flip channels are not CoP unless they reduce to one of the two classes above.

\subsection{QDPI under CoP channels: numerical illustration}
\label{sec:QDPI-numerical}

Since the conjectured CoP-QDPI has not yet been established analytically, we now present numerical evidence in support of its validity. The channel classification discussed above serves as a guide for selecting representative channels from the various CoP classes in both qubit and qutrit settings.

To investigate the conjectured QDPI numerically, we performed Monte Carlo simulations by generating $5\times10^4$ independent random pairs of density operators
$ (\rho,\sigma)\in\mathcal D(\mathcal H)\times \mathcal D(\mathcal H).$ For each channel-parameter regime, the pairs were sampled from the Ginibre ensemble and normalized to unit trace, subject to the constraint $[\rho,\sigma]\neq0.$ The corresponding regularized states $(\rho_\varepsilon,\sigma_\varepsilon)$, and all numerical evaluations were performed using these regularized states. Although Ginibre sampling produces full-rank states almost surely, the regularization procedure was retained throughout in order to remain fully consistent with the theoretical framework developed in this work.

The numerical tests were carried out for both qubit $\dim\mathcal H = 2$ and qutrit $ \dim\mathcal H = 3 $ systems. For parameterized noise channels, the channel parameter satisfies $p\in[0,1]$ by definition. In the numerical simulations, we considered the representative values $p \in \{0.0, 0.3, 0.5, 0.7, 0.9\}$. For clarity of presentation, only selected parameter values are displayed in the figures reported in the present work.

To distinguish genuine violations from numerical round-off effects, a data point was regarded as a violation only if
\begin{equation}
Q\left( \mathcal N(\rho_\varepsilon) \;\middle\|\; \mathcal N(\sigma_\varepsilon) \right) - Q(\rho_\varepsilon \| \sigma_\varepsilon) > \delta,
\end{equation}
where $\delta$ denotes a prescribed numerical tolerance. A violation threshold of $\delta = 10^{-10}$ was adopted for qubit systems and $\delta = 10^{-8}$ for qutrit systems. An exception was made for the qutrit dephasing channel. Since this channel lies outside the CoP class in general dimensions, a more stringent numerical test was performed by setting $\delta = 10^{-10}$ and increasing the sample size to $10^5$ random pairs in order to further suppress possible numerical artifacts.

In the scatter plots presented below, the horizontal axis represents the input quantity $Q(\rho_\varepsilon \| \sigma_\varepsilon)$, whereas the vertical axis represents the corresponding output quantity $Q\left( \mathcal N(\rho_\varepsilon) \;\middle\|\; \mathcal N(\sigma_\varepsilon) \right)$. The dashed red line corresponds to the equality relation
$
Q(\rho_\varepsilon \| \sigma_\varepsilon) = Q\left( \mathcal N(\rho_\varepsilon) \;\middle\|\; \mathcal N(\sigma_\varepsilon) \right).
$
For parameterized channels, the case $p=0$ reduces to the identity channel. Consequently, the relative quantumness remains unchanged under the channel action, and the sampled points accumulate near the equality line.

For nonzero noise parameters, the conjectured CoP-QDPI predicts $Q(\rho_\varepsilon \| \sigma_\varepsilon) \ge Q\left( \mathcal N(\rho_\varepsilon) \;\middle\|\; \mathcal N(\sigma_\varepsilon) \right)$.

Graphically, this implies that all admissible points should lie on or below the dashed equality line. Points located strictly below the line indicate a reduction of relative quantumness under the action of the channel, whereas points above the line would constitute numerical violations of the conjectured monotonicity relation.

\begin{center}
\textbf{Numerical results for qubit and qutrit CoP channels}
\end{center}

We begin with the qubit bit-flip channel, a representative example of a unital CoP map. Figure \ref{f1} shows the corresponding scatter plot for the representative parameter value $p=0.7$. All sampled points are located on or below the equality line, consistent with the conjectured inequality. The inset provides a closer view of the near-equality region and confirms the absence of any detectable violations. Within numerical precision, the bit-flip channel therefore exhibits behavior fully consistent with the proposed CoP-QDPI.
\begin{figure}[htbp!]
    \centering
    \includegraphics[width=0.8\linewidth]{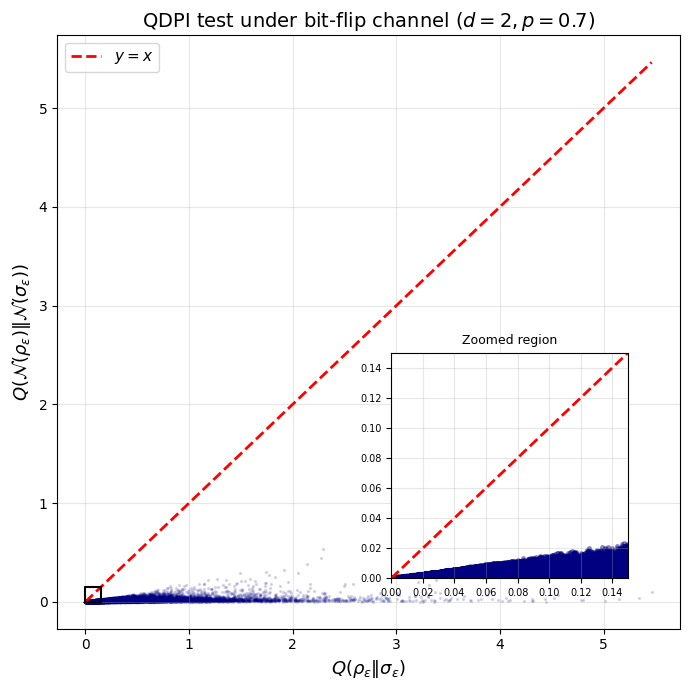}
    \caption{\small Scatter plot of $Q(\rho_\varepsilon\|\sigma_\varepsilon)$ versus $Q\!\left( \mathcal N(\rho_\varepsilon)\,\|\,\mathcal N(\sigma_\varepsilon) \right)$ for the qubit bit-flip channel with $p=0.7$. The dashed red line denotes the equality relation $Q(\rho_\varepsilon\|\sigma_\varepsilon)=Q\!\left(\mathcal N(\rho_\varepsilon)\,\|\,\mathcal N(\sigma_\varepsilon)\right)$. The inset shows a magnified view near the equality region.}
    \label{f1}
\end{figure}

We next consider the qutrit bit-flip channel. Unlike the qubit case, higher-dimensional Pauli-type flip channels are not, in general, commutativity-preserving and therefore lie outside the CoP framework discussed above. Figure \ref{f2} shows that a non-negligible fraction of the sampled points lies above the reference line, indicating violations of the monotonicity relation. This behavior is consistent with the fact that the qutrit bit-flip channel is not a CoP channel and therefore is not expected to satisfy the conjectured CoP-QDPI.

\begin{figure}[!htbp]
    \centering
    \includegraphics[width=0.8\linewidth]{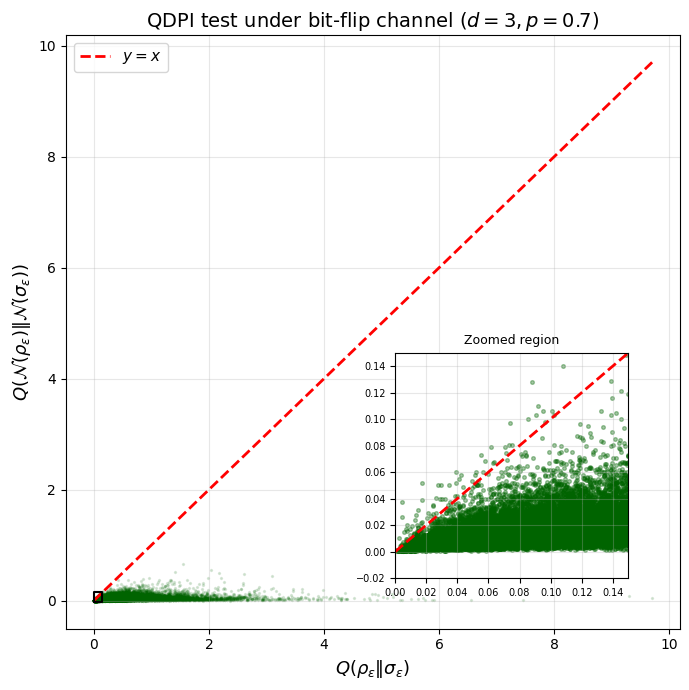}
    \caption{\small Scatter plot of $Q(\rho_\varepsilon\|\sigma_\varepsilon)$ versus $Q\!\left(\mathcal N(\rho_\varepsilon)\,\|\,\mathcal N(\sigma_\varepsilon)\right)$ for the qutrit generalized bit-flip channel with $p=0.7$. The dashed red line denotes the equality relation $Q(\rho_\varepsilon\|\sigma_\varepsilon)=Q\!\left(\mathcal N(\rho_\varepsilon)\,\|\,\mathcal N(\sigma_\varepsilon)\right)$. Points appearing above the reference line correspond to violations of the conjectured QDPI.}
    \label{f2}
\end{figure}

In contrast to the qutrit bit-flip channel, the depolarizing channel belongs to the CoP class in both qubit and qutrit dimensions. Consequently, it provides a natural test case for the conjectured CoP-QDPI beyond the qubit setting. As shown in Figure \ref{f3} and Figure \ref{f5} (Appendix \ref{App_B}), no sampled points are observed above the reference line in either dimension. All numerical data therefore remain consistent with the conjectured monotonicity relation.

\begin{figure}[!htbp]
    \centering
    \includegraphics[width=0.8\linewidth]{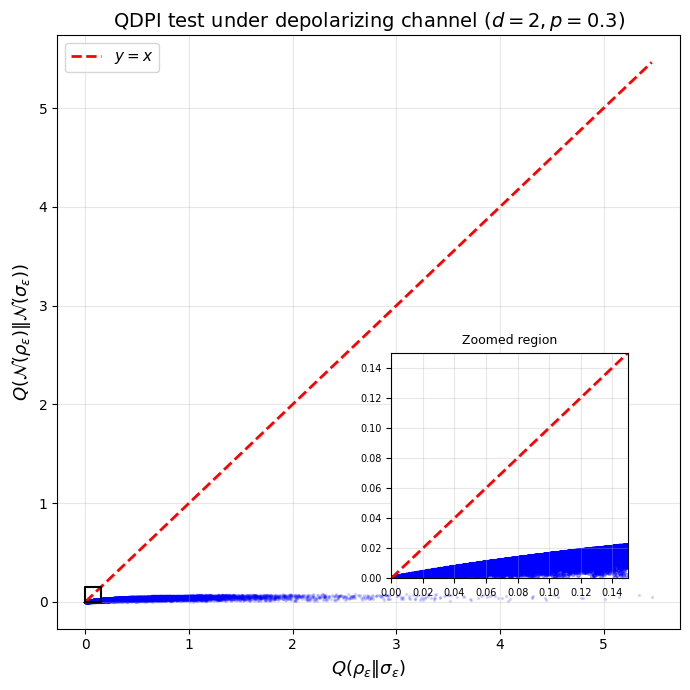}
    \caption{\small Scatter plot of $Q(\rho_\varepsilon\|\sigma_\varepsilon)$ versus $Q\!\left( \mathcal N(\rho_\varepsilon)\,\|\,\mathcal N(\sigma_\varepsilon) \right)$ for the qubit depolarizing channel with $p=0.3$. The dashed red line represents the equality relation. All sampled points remain below the reference line, supporting the conjectured QDPI.}
    \label{f3}
\end{figure}

The numerical evidence for the dephasing channel aligns with the QDPI conjecture in the qubit case. Extensive simulations were conducted across the entire range of noise parameters $p$ defined above for this experiment, and no violations were observed for any sampled pairs. For the sake of brevity, we present the representative results for $p=0.5$ in Figure \ref{f6} (Appendix \ref{App_B}), where no violations of the conjectured monotonicity relation are observed.

A different behavior is observed for the qutrit dephasing channel. While the majority of sampled pairs continue to satisfy the conjectured QDPI, occasional violations appear in the intermediate noise regime, particularly for $p=0.3,\,0.5,$ and $0.7.$

This observation is consistent with the known structure of CoP channels in higher dimensions. Unlike its qubit counterpart, the qutrit dephasing channel does not generally belong to the CoP class. Although dephasing suppresses off-diagonal coherence and tends to drive states toward a more classical structure, the non-uniform attenuation of different coherence sectors can still induce nontrivial changes in the relative noncommutativity structure between the two states after the channel action. The appearance of violations therefore provides additional evidence that the commutativity-preserving property is essential for the validity of the conjectured monotonicity relation. A representative example is shown in Figure \ref{f4}.

\begin{figure}[!htbp]
    \clearpage
    \centering
    \includegraphics[width=0.8\linewidth]{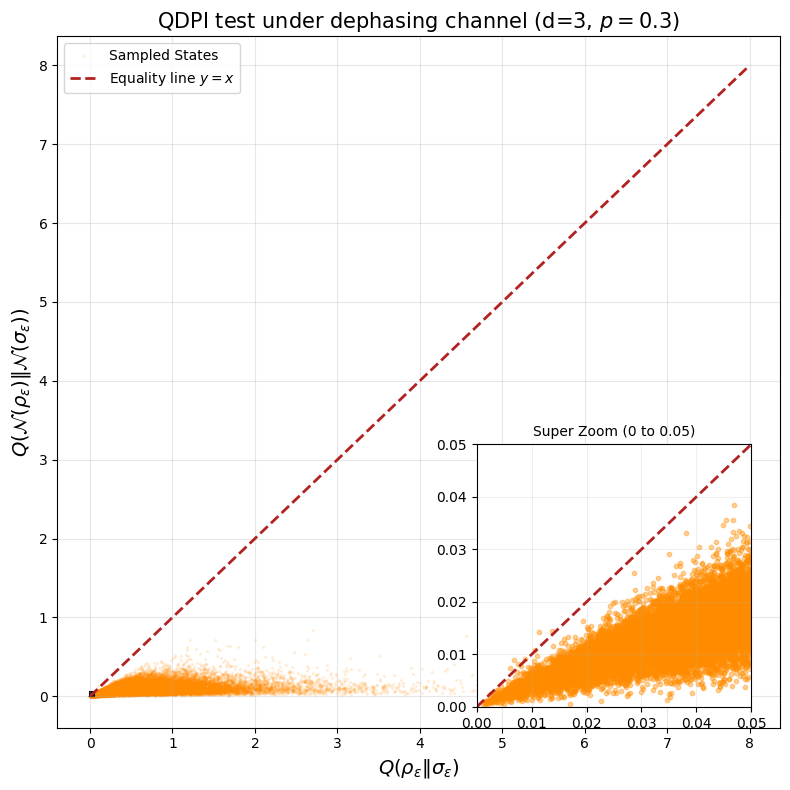}
    \caption{\small Scatter plot of $Q(\rho_\varepsilon\|\sigma_\varepsilon)$ versus $Q\!\left( \mathcal N(\rho_\varepsilon)\,\|\,\mathcal N(\sigma_\varepsilon) \right)$ for the qutrit dephasing channel with $p=0.3$. The dashed red line represents the equality relation. While the majority of sampled points remain below the reference line, a small number of violations are observed in the qutrit case. The inset shows a magnified view of the near-equality region.}
    \label{f4}
\end{figure}

We next examine the phase-flip channel. In the qubit setting, this channel belongs to the CoP class, and no violations of the conjectured QDPI are observed across all tested values of $p$. The numerical results therefore remain fully consistent with the conjectured monotonicity of the regularized relative quantumness. A representative scatter plot for $p=0.7$ is shown in Figure \ref{f7} (Appendix \ref{App_B}).

Similarly to the bit-flip channel, higher-dimensional phase-flip channels are not generally CoP. Consistent with this structural distinction, the numerical results show that the regularized relative quantumness is not monotonically decreasing under these channels. Except for the trivial case $p=0$, violations of the proposed monotonicity relation are observed throughout the tested parameter regimes. A representative qutrit example is presented in Figure \ref{f8} (Appendix \ref{App_B}).

Two additional classes of CoP channels considered in this work are the semi-classical channels and transpose-type isotropic channels. Unlike higher-dimensional bit-flip and phase-flip channels, these maps belong explicitly to the CoP classifications in both qubit and qutrit dimensions. They therefore provide particularly relevant test cases for the conjectured CoP-QDPI.

For both qubit and qutrit systems, the semi-classical channel exhibits a pronounced suppression of the relative quantumness. As shown in Figure \ref{f9} and Figure \ref{f10} (Appendix \ref{App_B}), the output values collapse close to the horizontal axis, reflecting the strongly decohering nature of the channel. No violations were observed within numerical precision.

The transpose-type isotropic channel exhibits the same qualitative behavior in both dimensions. Detailed results are shown in Figure \ref{f11} and Figure \ref{f12} (Appendix \ref{App_B}). Although the output values display a broader distribution than in the semi-classical case, all sampled points remain below the reference line. These results provide further numerical support for the conjectured monotonicity of the regularized relative quantumness under CoP channels. Additional simulations across different admissible values of the parameter $t$ within the CPTP region produced qualitatively similar behavior, with no observed violations of the conjectured inequality.

Taken together, the numerical results reveal a clear pattern. Channels belonging to the CoP class consistently satisfy the proposed monotonicity relation within numerical precision, whereas violations arise for channels that fall outside the CoP framework. This provides strong numerical evidence supporting the conjectured CoP-QDPI and highlights the central role of commutativity preservation in governing the behavior of the regularized relative quantumness.

\subsection{Preliminary numerical evidence from prior work} 

The numerical investigation summarized here is distinct from the numerical illustration for CoP channels in Section \ref{CoP-QDPI}. A preliminary study of QDPI for general CPTP maps and arbitrary $\alpha\in(0,\infty)\setminus\{1\}$, a regime in which an analytic proof remains open, was reported in\cite{MeunsonDeesuwan2025CQRRF}. The study covered four representative channels: the generalized noisy identity channel, the depolarizing channel, the dephasing channel, and the amplitude damping channel. Across $10^4$ randomly generated qubit-state pairs per run over ten independent runs, no violations of QDPI were observed within $\alpha\in(0.9,1.3)$ for all tested channels and noise strengths.

\section*{Conclusion}

We introduced the Cu-Q relative Rényi functional, a cumulant-based quantum divergence derived from the quantum relative surprisal and motivated by the cumulant-generating function framework. The proposed construction extends the classical connection between Rényi divergence and statistical cumulants to the quantum setting and admits a path-integral-like representation.

On its natural domain, we established several fundamental properties expected of a Rényi-type divergence, including positivity, reduction to the classical case, additivity, unitary invariance, continuity, and monotonicity with respect to the Rényi parameter. Through a regularized formulation at $\alpha=0$, we further showed that the associated relative quantumness quantity vanishes if and only if the underlying states commute, thereby providing a complete characterization of non-commutativity.

The validity of the quantum data-processing inequality remains open. Nevertheless, numerical investigations reveal a clear distinction between commutativity-preserving and non-commutativity-preserving channels. In particular, no violations were observed for the CoP channels considered in this work, providing strong evidence for the conjectured CoP-QDPI.

Future work will focus on establishing rigorous QDPI results, extending the numerical analysis to broader classes of CPTP maps, and exploring potential applications in quantum thermodynamics and related areas.
\section*{Acknowledgements} \label{sec:acknowledgements}
    This work was supported by the National Science, Research and Innovation Fund (NSRF) through the Program Management Unit for Human Resources and Institutional Development, Research and Innovation under grant number B39G680007.


\appendix 

\section{Derivation of the Cu-Q Relative Rényi Functional}
\label{app:derivation_cuq}

In this appendix, we show how the proposed Cu-Q relative Rényi functional arises naturally from the cumulant-generating-function (CGF) formalism. The derivation begins with the classical theory of statistical fluctuations and culminates in a non-commutative extension based on the quantum relative surprisal operator.

\subsection{Rényi entropy and cumulant-generating functions}

A fundamental tool in probability theory and statistics is the cumulant-generating function (CGF). For a real-valued random variable $Y$, the CGF is defined as
\begin{equation}
\mathcal K_Y(\zeta):=\ln\mathbb E\left[ e^{\zeta Y}\right],\quad\zeta\in\mathbb R.
\end{equation}
Its importance stems from the fact that the Taylor expansion around $\zeta=0$ generates the cumulants of $Y$:
\begin{equation}
\mathcal K_Y(\zeta)=\sum_{n=1}^{\infty} \frac{\zeta^n}{n!}
C_n(Y),
\end{equation}
where
\begin{equation}
C_n(Y)=\left.\frac{d^n}{d\zeta^n}\mathcal K_Y(\zeta)\right|_{\zeta=0}.
\end{equation}
The first cumulant is the mean, the second is the variance, while higher-order cumulants characterize increasingly refined statistical fluctuations such as skewness and kurtosis.

In information theory, the fundamental random quantity associated with an event $X=x$ is the surprisal $\Lambda(X)=-\ln p(X)$, which quantifies the information gained upon observing the outcome $X$. Since surprisal itself is a random variable, it is natural to study its statistical fluctuations through the CGF. Taking $Y=\Lambda(X)$ gives
\begin{equation}
\mathcal K_{\Lambda(X)}(\zeta)=\ln\mathbb E\left[e^{\zeta\Lambda(X)}\right].
\end{equation}
Substituting the definition of the surprisal yields
\begin{align}
\mathcal K_{\Lambda(X)}(\zeta)
&=
\ln\sum_x p(x) e^{-\zeta\ln p(x)}\nonumber\\
&=
\ln\sum_xp(x)^{1-\zeta}.
\end{align}
Evaluating the CGF at $\zeta = 1-\alpha$ gives
\begin{equation}
\mathcal K_{\Lambda(X)}(1-\alpha)=\ln\sum_x
p(x)^\alpha.
\end{equation}
Consequently,
\begin{equation}
S_\alpha(X)=
\frac{1}{1-\alpha}
\mathcal K_{\Lambda(X)}(1-\alpha),
\end{equation}
which is precisely the classical Rényi entropy.

This observation provides a natural statistical interpretation of Rényi entropy. Rather than being introduced as an independent information measure, Rényi entropy emerges as a normalized CGF of the surprisal random variable, evaluated at $\zeta=1-\alpha$. Consequently, it depends not only on the mean surprisal but also on the higher-order cumulants that characterize statistical fluctuations. This viewpoint serves as the conceptual starting point for the construction of the Cu-Q relative Rényi
functional developed in this work.

\subsection{Classical Rényi divergence and relative surprisal}

The observation that Rényi entropy can be expressed as a normalized cumulant-generating function of the surprisal suggests that a similar statistical interpretation should also exist in the relative setting. In particular, since Rényi entropy arises from the CGF of $\Lambda(X)=-\ln p(X)$, it is natural to ask whether the classical Rényi divergence can likewise be represented through the CGF of an appropriate relative surprisal quantity.

To investigate this question, we consider two probability distributions $p$ and $q$ and the relative surprisal $\Delta\Lambda(x)=\ln p(x)-\ln q(x)$. Starting from the classical Rényi divergence,
\begin{equation}
S_\alpha(p\|q) =\frac{1}{\alpha-1}\ln\sum_x p(x)^\alpha
q(x)^{1-\alpha},
\end{equation}
we rewrite the quantity inside the logarithm as
\begin{align}
\sum_x p(x)^\alpha q(x)^{1-\alpha}
&=
\sum_x p(x)\left(\frac{p(x)}{q(x)} \right)^{\alpha-1}
\nonumber\\
&=
\sum_xp(x)e^{(\alpha-1)\Delta\Lambda(x)}.
\end{align}
Recognizing the last expression as an expectation value with respect to
$p$, we obtain
\begin{equation}
\sum_x
p(x)^\alpha
q(x)^{1-\alpha}
=
\mathbb E_p
\!\left[
e^{(\alpha-1)\Delta\Lambda(X)}
\right].
\end{equation}
Hence, we can define the CGF of the relative surprisal by
\begin{equation}
\mathcal K_{\Delta\Lambda(X)}(\tau)
:=
\ln
\mathbb E_p
\!\left[
e^{(\tau)\Delta\Lambda(X)}
\right].
\end{equation}
Consequently,
\begin{equation}
S_\alpha(p\|q)
=
\frac{1}{\alpha-1}
\mathcal K_{\Delta\Lambda(X)}
(\alpha-1),
\end{equation}
showing that the classical Rényi divergence is precisely a normalized CGF of the relative surprisal evaluated at $\tau=\alpha-1$.

\subsection{Quantum extension}

Since the classical Rényi divergence can be viewed as a normalized cumulant-generating function of the relative surprisal, it is natural to seek a quantum analogue based on the quantum relative surprisal
operator $\Delta\Xi= \ln\rho-\ln\sigma$. We therefore introduce the associated non-commutative
cumulant-generating functional
\begin{equation}
\mathcal F_{\rho,\sigma}(\tau)
:=
\ln
\operatorname{Tr}
\!\left[
\rho
e^{\tau\Delta\Xi}
\right].
\end{equation}

Following the same normalized cumulant-generating-function principle as in the classical case, we define the Cu-Q relative Rényi functional by
\begin{equation}
S_\alpha^Q(\rho\|\sigma)
=
\frac{1}{\alpha-1}
\mathcal F_{\rho,\sigma}(\alpha-1).
\end{equation}

Equivalently,
\begin{equation}
S_\alpha^Q(\rho\|\sigma)
=
\frac{1}{\alpha-1}
\ln
\operatorname{Tr}
\!\left[
\rho
e^{(\alpha-1)(\ln\rho-\ln\sigma)}
\right].
\end{equation}

Therefore, the proposed Cu-Q relative Rényi functional arises naturally from the cumulant-generating function of the quantum relative surprisal, extending the classical Rényi-cumulant correspondence to
the non-commutative setting. This provides the fundamental motivation for the construction and suggests that the resulting quantity encodes
not only the mean quantum relative surprisal but also its higher-order statistical fluctuations through a hierarchy of non-commutative
cumulants.

\onecolumngrid 
\section{Additional Figures}
\label{App_B}

This appendix collects additional scatter plots for the channels and dimensions discussed in Section~\ref{sec:QDPI-numerical}.

\begin{figure}[htbp]
    \centering
    \begin{minipage}{0.48\textwidth}
        \centering
        \includegraphics[width=\linewidth]{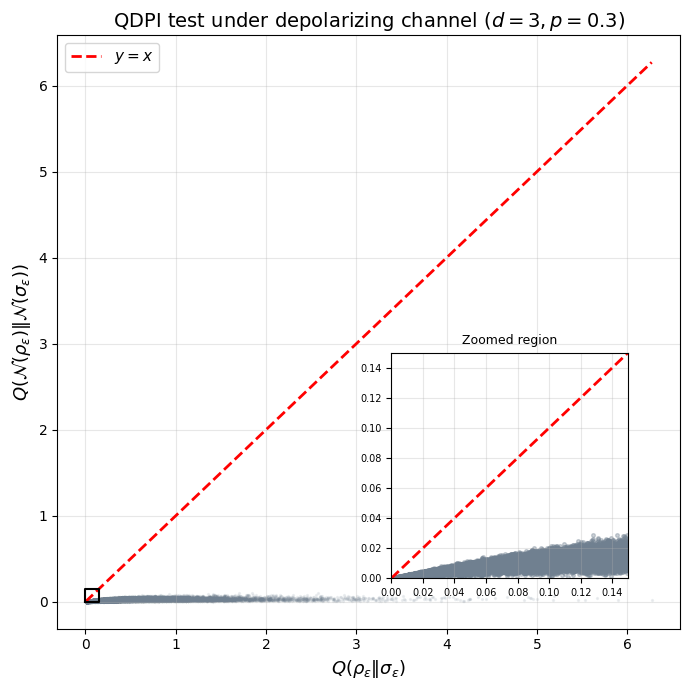}
        \caption{Qutrit depolarizing channel ($d=3$, $p=0.3$).}
        \label{f5}
    \end{minipage}
    \hfill
    \begin{minipage}{0.48\textwidth}
        \centering
        \includegraphics[width=\linewidth, trim=0 5 0 0, clip]{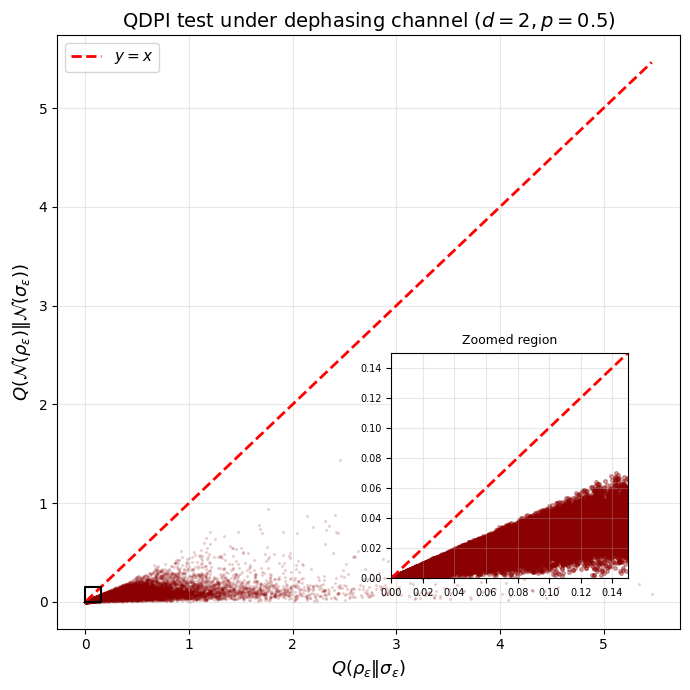}
        \caption{Qubit dephasing channel ($d=2$, $p=0.5$).}
        \label{f6}
    \end{minipage}

    \vspace{1cm}

    \begin{minipage}{0.48\textwidth}
        \centering
        \includegraphics[width=\linewidth]{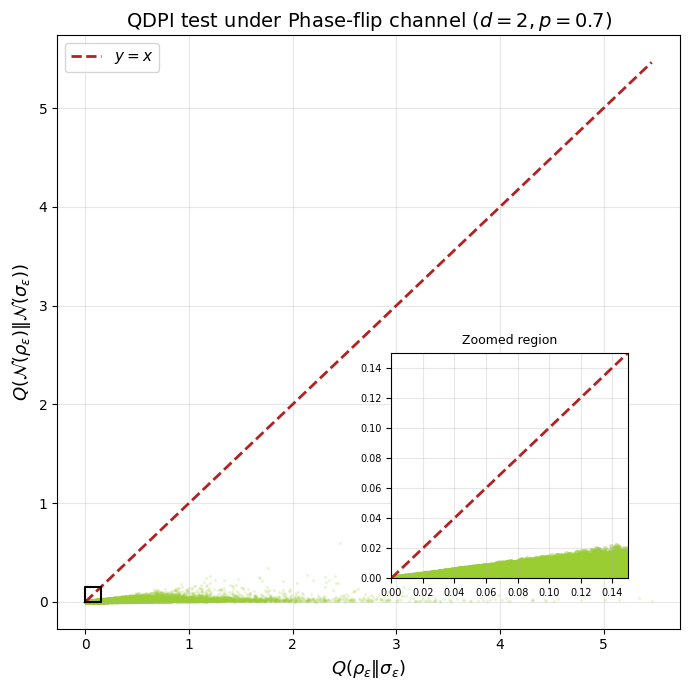}
        \caption{Qubit phase-flip channel ($d=2$, $p=0.7$).}
        \label{f7}
    \end{minipage}
    \hfill
    \begin{minipage}{0.48\textwidth}
        \centering
        \includegraphics[width=\linewidth, trim=0 5 0 0, clip]{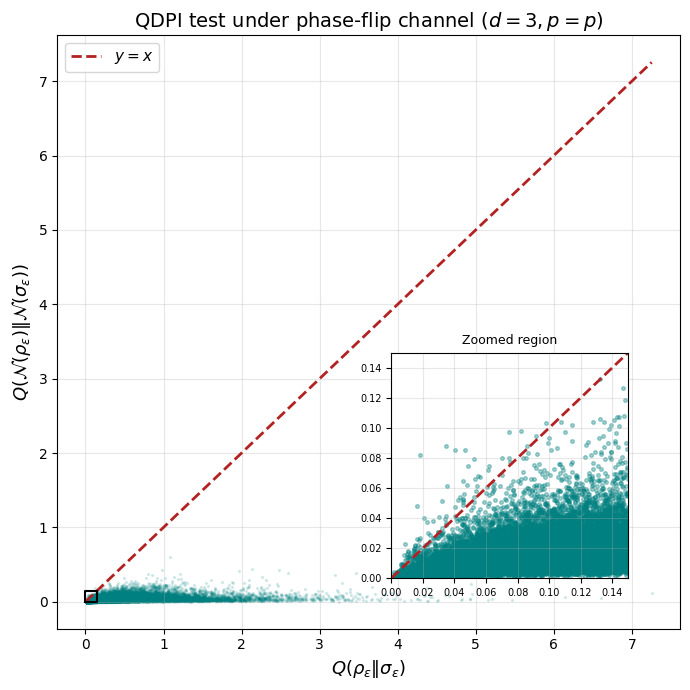}
        \caption{Qutrit phase-flip channel ($d=3$, $p=0.7$).}
        \label{f8}
    \end{minipage}
\end{figure}

\clearpage

\begin{figure}[htbp]
    \centering
    \begin{minipage}{0.48\textwidth}
        \centering
        \includegraphics[width=\linewidth]{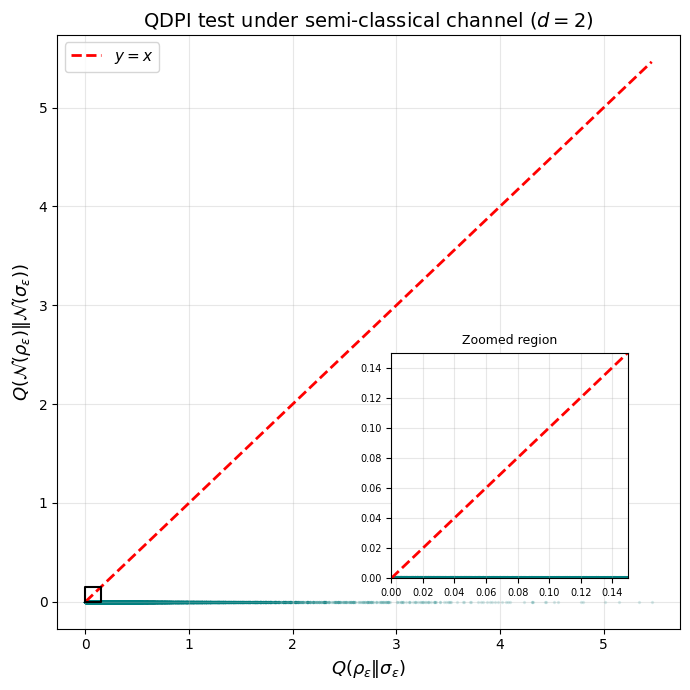}
        \caption{Qubit semi-classical channel ($d=2$).}
        \label{f9}
    \end{minipage}
    \hfill
    \begin{minipage}{0.48\textwidth}
        \centering
        \includegraphics[width=\linewidth]{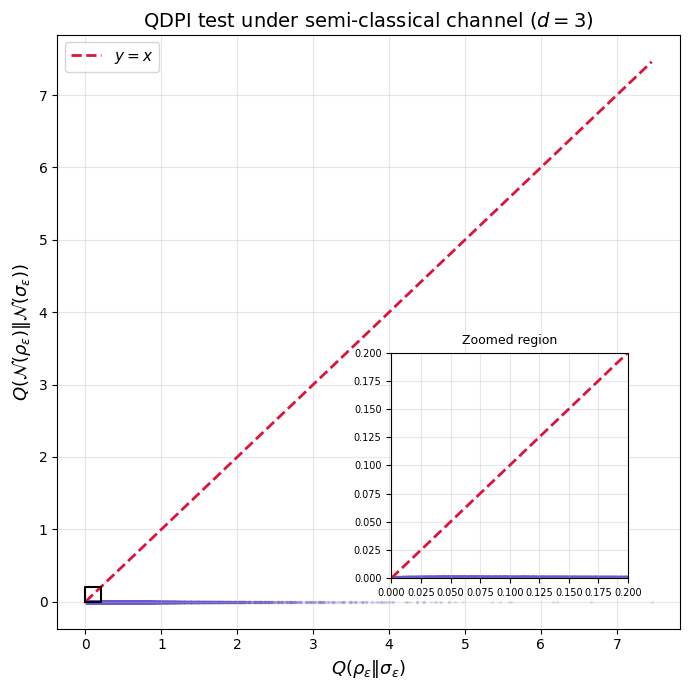}
        \caption{Qutrit semi-classical channel ($d=3$).}
        \label{f10}
    \end{minipage}

    \vspace{1cm}

    \begin{minipage}{0.48\textwidth}
        \centering
        \includegraphics[width=\linewidth]{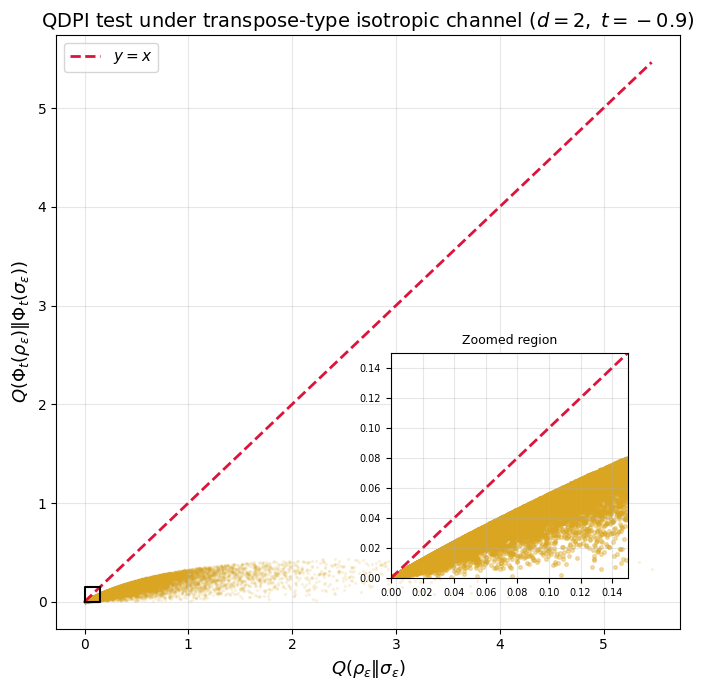}
        \caption{Qubit transpose-type isotropic channel ($d=2$, $t=-0.9$).}
        \label{f11}
    \end{minipage}
    \hfill
    \begin{minipage}{0.48\textwidth}
        \centering
        \includegraphics[width=\linewidth]{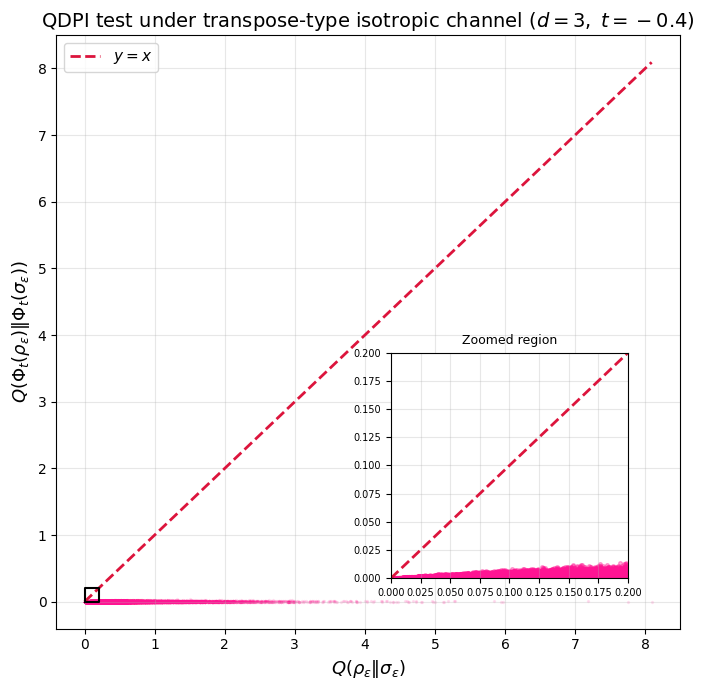}
        \caption{Qutrit transpose-type isotropic channel ($d=3$, $t=-0.4$).}
        \label{f12}
    \end{minipage}
\end{figure}

\twocolumngrid
\clearpage 
\twocolumngrid 
\bibliographystyle{apsrev4-1}  
\bibliography{references}
\end{document}